\documentclass{article}
\usepackage{amssymb,mathtools,import,geometry,graphicx,amsthm,centernot}
\usepackage[colorlinks=false]{hyperref}
      \usepackage{setspace}

\newcommand{\N}{\ensuremath{\mathbb{N}} }

\newcommand{\fbins}{\ensuremath{\{0,1\}^*}}
\newcommand{\infbins}{\ensuremath{\{0,1\}^{\omega}}}
\newcommand{\ILUPDC}{\ensuremath{\mathrm{ILUPDC}}}
\newcommand{\bin}{\ensuremath{\{0,1\}}}
\newcommand{\finfbins}{\ensuremath{\bin^{\leq \omega}}}
\newcommand{\rev}[1]{\ensuremath{#1^{-1}}}
\newcommand{\thh}{\ensuremath{\textrm{th}}}
\newcommand{\ILFST}{\ensuremath{\mathrm{ILFST}}}
\newcommand{\binf}{\ensuremath{\mathrm{bin}}}
\newcommand{\stringf}{\ensuremath{\mathrm{string}}}
\renewcommand{\notin}{\ensuremath{\centernot\in}}
\newcommand{\FSTsize}[1]{\ensuremath{\mathrm{FST}^{\leq #1}}}
\newcommand{\dimFS}{\ensuremath{\mathrm{dim}_{\mathrm{FS}}}}
\newcommand{\DimFS}{\ensuremath{\mathrm{Dim}_{\mathrm{FS}}}}
\newcommand{\FScomp}[1]{\ensuremath{D_{\mathrm{FS}}^{#1}}}

\newcommand{\ILPDC}{\ensuremath{\mathrm{ILPDC}}}
\newcommand{\LZ}{\ensuremath{\mathrm{LZ}}}
\theoremstyle{plain}
\newtheorem{theorem}{Theorem}[section]
\newtheorem{corollary}[theorem]{Corollary}
\newtheorem{lemma}[theorem]{Lemma}
\theoremstyle{definition}
\newtheorem{definition}[theorem]{Definition}

\newtheorem{remark}[theorem]{Remark}

\begin{document}

\title{Pushdown and Lempel-Ziv Depth}

\author{Liam Jordon\thanks{Supported by a postgraduate scholarship from the Irish Research Council.}\\
\textrm{liam.jordon@mu.ie} \\
\and 
Philippe Moser \\
\textrm{pmoser@cs.nuim.ie}}

\date{%
    Dept. of Computer Science, Maynooth University, Maynooth, Co. Kildare, Ireland\\[2ex]%
}


\maketitle
    \begin{abstract}
        This paper expands upon existing and introduces new formulations of Bennett's logical depth. In previously published work by Jordon and Moser, notions of finite-state depth and pushdown depth were examined and compared. These were based on finite-state transducers and information lossless pushdown compressors respectively. Unfortunately a full separation between the two notions was not established. This paper introduces a new formulation of pushdown depth based unary-stack pushdown compressors. This improved formulation allows us to do a full comparison by demonstrating the existence of sequences with high finite-state depth and low pushdown depth, and vice-versa. A new notion based on the Lempel-Ziv 78 algorithm is also introduced. Its difference from finite-state depth is shown by demonstrating the existence of a Lempel-Ziv deep sequence that is not finite-state deep and vice versa. Lempel-Ziv depth's difference from pushdown depth is shown by building sequences that have a pushdown depth level of roughly $1/2$ but low Lempel-Ziv depth, and a sequence with high Lempel-Ziv depth but low pushdown depth. Properties of all three notions are also discussed and proved.
    \end{abstract}

  \section{Introduction}

In 1988 Charles Bennett introduced a new method to measure the \textit{useful} information contained in a piece of data \cite{b:bennett88}. This measurement tool is called \emph{logical depth}. Logical depth helps to formalise the difference between complex and non-complex structures. Intuitively, deep structures can be thought of as structures that contain patterns which are incredibly difficult to find. Given more and more time and resources, an algorithm could spot these patterns and exploit them (such as to compress a sequence). Non-deep structures are sometimes referred to as being shallow. Random structures are not considered deep as they contain no patterns. Simple structures are not considered deep as while they contain patterns, they are too easy to spot.

Bennett's original notion is based on Kolmogorov complexity \cite{b:bennett88,DBLP:journals/tcs/JuedesLL94}, and interacts nicely with fundamental notions of computability theory \cite{DBLP:journals/tcs/DowneyMN17,DBLP:journals/dmtcs/MoserS17}. Due to the uncomputabilty of Kolmogorov complexity, several researchers have attempted to adapt Bennett's notion to lower complexity levels. While variations have been based on computable notions \cite{DBLP:journals/iandc/LathropL99}, more feasible notions based on polynomial time computations \cite{b.antunes.depth.journal,DBLP:journals/tcs/Moser13,DBLP:journals/iandc/Moser20} have been studied, including both finite-state transducers and lossless pushdown compressors \cite{DotyM07,DBLP:conf/sofsem/JordonM20}. Similarly to randomness, there is no absolute notion of logical depth. Although logical depth was originally defined as depending on both computational and minimal descriptional complexity, the notions in this paper are focused purely on a minimal descriptional length based complexity, specifically the ratio between the length of the input and the length of the output to restricted classes of transducers and compression algorithms. We use the term depth to describe our notion as all variants mentioned above can be seen as variations of a same theme based on the compression framework \cite{DBLP:journals/tcs/Moser13}.        

While this switching from the original Kolmogorov complexity based definition of logical depth to more feasible notions does lead to a trade-off between some properties of the original notion, most notions satisfy some basic properties that could be seen as fundamental. These are:                               
\begin{itemize}                                                               
\item Random sequences are not deep (for the appropriate randomness notion).                                                   
\item Computable sequences are not deep (for the appropriate computability notion).                                            
\item A slow growth law: deep sequences cannot be \emph{quickly} computed from shallow ones.                                   
\item Deep sequences exist.                                                                                        
\end{itemize}

In this paper we continue the study of depth via classes of automaton and compression algorithms. For two families of compression algorithms $T$ and $T'$, we say a sequence is $(T,T')$-deep if for every compressor $C$ of type $T$, there exists a compressor $C'$ of type $T'$  such that on almost every prefix of $S$ (with length denoted $n$), $C'$ compresses it at least by $\alpha n$ more bits than $C$, for some constant $\alpha$. We refer to $\alpha$ as the $(T,T')$-depth level of $S$. We drop the $(T,T')$ notation and refer to just $T$ or $T'$ depending on the context when referring to the depth notions discussed in this paper.

Doty and Moser first presented an infinitely often notion of depth based on finite-state transducers in \cite{DotyM07} based on the minimal length of an input to a finite-state transducer that results in the desired output. Further study of finite-state minimal descriptional length can be found in \cite{calude2011finite,calude2016finite}. This led to Jordon and Moser introducing a notion based on lossless pushdown compressors in \cite{DBLP:conf/sofsem/JordonM20} where it was shown that there existed a finite-state deep sequence which was not pushdown deep. The contrary result was not established, i.e. the existence of a pushdown deep sequence that is not finite-state deep. In this paper we present a new notion of pushdown depth which provides a much clearer separation from finite-state depth by allowing us to prove the existence of sequences which are pushdown deep and if they are finite-state deep, have a very low level of depth, and vice versa. This notion of pushdown depth is based on the output of information lossless pushdown compressors (ILPDCs) for a given input. The model of pushdown compressors used is found in \cite{DBLP:journals/mst/MayordomoMP11}. Specifically we examine the difference in compression of ordinary ILPDCs, which have the ability to push two symbols onto their stacks, against ILPDCs which can only push one symbol onto their stacks. These ILPDCs are referred to as information lossless unary-stack pushdown compressors (ILUPDCs).

We also introduce a new notion called Lempel-Ziv depth (LZ-depth) based on the Lempel-Ziv 78 (LZ) compression algorithm introduced in \cite{DBLP:LZ78}. LZ-depth examines the output of a lossless finite-state transducer against the output of the LZ algorithm on a given input. 

For both the pushdown and Lempel-Ziv depth notions, we demonstrate that each notion a subset of some of the fundamental depth properties, i.e. both easy and random sequences based on the setting are not deep, and that a slow growth law holds. 

When comparing the three notions we examine the depth level of various sequences. To compare pushdown depth with finite-state depth, we first show the existence of an i.o. finite-state deep sequence which is not PD-deep. We then show the existence of a PD-deep sequence with depth level of roughly $1/2$ such that, if it is finite-state deep, it has a low finite-state depth level. To compare finite-state depth and LZ-depth we show that there exists a normal sequence that is LZ-deep. Since no normal sequence is finite-state deep \cite{DotyM07}, this demonstrates a difference. We also build a sequence which is finite-state deep and infinitely often LZ-deep, but not almost everywhere LZ-deep. When comparing pushdown depth with LZ-depth, we show that there exists a sequence that is not PD-deep but is LZ-deep. We then build a sequence which is PD-deep and has a PD-depth level of roughly $1/2$, but low LZ-depth.

\section{Preliminaries}

\subsection{Notation}
We work with a binary alphabet and all logarithms are taken in base $2$. A (finite) \textit{string} is an element of $\fbins$. $\bin^n$ denotes the set of strings of length $n$. Thus $\fbins = \bigcup_{n=1}^\infty \bin ^n$. Strings will generally be denoted by a lowercase letter. The set of (infinite) \textit{sequences} is denoted by $\infbins$. Sequences will generally be denoted by an uppercase letter. $\finfbins = \fbins \cup \infbins$ denotes the set of all strings and sequences. $|x|$ denotes the length of the string $x$. We say $|S| = \omega$ for a sequence $S \in \infbins$. We use $\lambda$ to denote the empty string (the string of length $0$). For $x \in \fbins$ and $y \in \bin^{\leq \omega}$, $xy$ (occasionally written as $x \cdot y)$ denotes the string (or sequence) of $x$ concatenated with $y.$ For $x \in \fbins$, $x^n$ denotes the string of $x$ concatenated with itself $n$ times, i.e. $x^n = \overbrace{x\cdot x \cdots x}^{n \textrm{ times}}.$ For $x \in \finfbins$ and $0\leq i < |x|$, $x[i]$ denotes the $i^{\thh}$ character of $x$ with the leftmost character being $x[0]$. For $x \in \finfbins$ and $0\leq i \leq j < |x|$, $x[i..j]$ denotes the \textit{substring} of $x$ consisting of the $i^{\thh}$ through $j^{\thh}$ bits of $x$. For $x \in \fbins$ and $y,z \in \finfbins$ such that  $z = xy,$ we call $x$ a \textit{prefix} of $z$ and $y$ a suffix of $z$. We write $x \sqsubseteq v$ to denote that $x$ is a prefix of $v$ and $x \ \sqsubset v$ is $x$ is a prefix of $v$ but $x \neq v$. For $x \in \finfbins$ we write $x \upharpoonright n$ to denote the prefix of length $n$ of $x$, i.e. $x \upharpoonright n = x[0..n-1].$ For $x \in \fbins,$ we write $\rev{x}$ to denote the reverse of $x$, that is, if $x = x_1\ldots x_m$, then $\rev{x} = x_m \ldots x_1$. For $x \in \fbins,$ we use $d(x)$ to denote the string constructed by doubling every bit of $x$. That is, if $x = x_1\ldots x_m$, $d(x) = x_1x_1 \ldots x_mx_m.$ 

 We write $K(x)$ to represent the plain Kolmogorov complexity of string $x$. That is, for a fixed universal Turing machine $U$, $$K_U(x) = \min \{ |y| : y \in \fbins, U(y) = x \}.$$ That is, $y$ is the shortest input to $U$ that results in the output of $x$.  The value $K_U(x)$ does not depend on the choice of universal machine up to an additive constant, therefore we drop the $U$ from the notation. Other authors commonly use $C$ to denote plain complexity (see \cite{downey:book,nies:book}), however we reserve $C$ to denote compressors. Note that for all $n \in \N$, there exists a string $x \in \bin^n$ such that $K(x) \geq |x|$ by a simple counting argument.
 
 We use Borel normality \cite{borelNormal} to examine the properties of some sequences. We say that a sequence $S$ is \emph{normal} if for all strings $x\in \fbins$, $x$ occurs with asymptotic frequency $2^{-|x|}$ as a substring in $S$.
 
 The following are the main two ways we examine the complexity of sequences. For a sequence $S$ and a function $C: \fbins \rightarrow \fbins$ the $C$-upper and lower compression ratio of $S$ are given by
 \begin{align*}
    \rho_{C}(S) &=  \liminf\limits_{n \to \infty} \frac{|C(S \upharpoonright n)|}{n}, \mathrm{\, and} \\
    R_{C}(S) &=  \limsup\limits_{n \to \infty} \frac{|C(S \upharpoonright n)|}{n}.
\end{align*}
For a sequence $S$ and a set $T$ of functions from $\fbins$ to  $\fbins$,  the $T$-\emph{best case} and $T$-\emph{worst case} compression ratios of $S$ are given by
 \begin{align*}
    \rho_{T}(S) &=  \inf \{ \rho_C : C \in T \}, \mathrm{ \,and}\\
    R_{T}(S) &=  \inf \{ R_C : C \in T \}.\\
    \end{align*}

    \subsection{Finite-State Transducers}
\label{FS: Sec: FST}

We use the standard finite-state transducer model.

\begin{definition}
 A \emph{finite-state transducer (FST)} is a $4$-tuple $T = (Q,q_0,\delta,\nu)$, where
\begin{itemize}
    \item $Q$ is a non-empty, finite set of \emph{states},
    \item $q_0 \in Q$ is the \emph{initial state},
    \item $\delta : Q\times \bin \rightarrow Q$ is the \emph{transition function}, 
    \item $\nu : Q\times \bin \rightarrow \fbins$ is the \emph{output function}.
\end{itemize}
\label{FS: Def: FST}
\end{definition}

For all $x \in \fbins$ and $b \in \bin$, the \emph{extended transition function} $\widehat{\delta}:q_0 \times \fbins \rightarrow Q$ is defined by the recursion $\widehat{\delta}(\lambda) = q_0$ and $\widehat{\delta}(xb) = \delta(\widehat{\delta}(x),b).$ For $x \in \fbins$, the output of $T$ on $x$ is the string $T(x)$ defined by the recursion $T(\lambda) = \lambda$, and $T(xb) = T(x)\nu(\widehat{\delta}(x),b)$. We require the class of \textit{information lossless} finite-state transducers to later demonstrate a slow growth law.

\begin{definition}
An FST $T$ is \emph{information lossless (IL)} if for all $x \in \fbins$, the function $x \mapsto (T(x), \widehat{\delta}(x))$ is injective. 
\label{FS: Def: ILFST}
\end{definition}

In other words, an FST $T$ is IL if the output and final state of $T$ on input $x$ uniquely identify $x$. We call an FST that is IL an ILFST. By the identity FST, we mean the ILFST $I_{\mathrm{FS}}$ that on every input $x \in \fbins$, $I_{\mathrm{FS}}(x) = x.$ We write (IL)FST to denote the set of all (IL)FSTs. We note that occasionally we call ILFSTs \textit{finite-state compressors} to emphasise when we view the ILFSTs as compressors as opposed to decompressors.

We require the concept of \textit{(information lossless) finite-state computable} functions to demonstrate our slow growth also.

\begin{definition}
A function $f: \infbins \rightarrow \infbins$ is said to be  \emph{(information lossless) finite-state computable ((IL)FS computable)} if there is an (IL)FST $T$ such that for all $S \in \infbins$, it holds that $\lim\limits_{n \to \infty}|T(S \upharpoonright n)| = \infty$ and for all $n \in \N$, $T(S \upharpoonright n) \sqsubseteq f(S)$.
\label{FS: Def: ILFS Computable}
\end{definition}
Based on the above definition, if $f$ is (IL)FS computable via the (IL)FST $T$, we say that $T(S) = f(S)$.
We often use the following two results \cite{Huff59a,Koha78} that demonstrate that any function computed by an ILFST can be inverted to be approximately computed by another ILFST. 

\begin{theorem}[\cite{Huff59a,Koha78}]
 For all $T \in \ILFST$, there exists $T^{-1} \in \ILFST$ and a constant $c \in \N$ such that for all $x \in \fbins$, $x \upharpoonright (|x| - c) \sqsubseteq T^{-1}(T(x)) \sqsubseteq x$.
 \label{FS: Thm: inverse fst}
\end{theorem}

\begin{corollary}
 For all $T \in \ILFST$, there exists $T^{-1} \in \mathrm{ILFST}$ such that for all $S \in \infbins$, $T^{-1}(T(S)) = S.$
\end{corollary}

\subsection{Pushdown Compressors}
The model of pushdown compressors (PDC) we use to define pushdown depth can be found in \cite{DBLP:journals/mst/MayordomoMP11} where PDCs were referred to as \textit{bounded pushdown compressors}. We use this model  as it allows for feasible run times by bounding the number of times a PDC can pop a bit from its stack without reading an input bit. This prevents the compressor spending an arbitrarily long time altering its stack without reading its input. This model of pushdown compression also has the nice property that it equivalent to a notion of pushdown-dimension based on bounded pushdown gamblers \cite{DBLP:conf/birthday/AlbertMM17}. Similar models where there is not bound on the number of times a bit can be popped off from the stack can be found in \cite{DBLP:journals/tcs/DotyN07}. 

The following contains details of the model used. It is taken from \cite{DBLP:journals/mst/MayordomoMP11}.
\begin{definition}
A \emph{pushdown compressor (PDC)} is a 7-tuple $C = (Q, \Gamma, \delta, \nu, q_0, z_0,c)$ where
\begin{enumerate}
    \item $Q$ is a non-empty, finite set of \textit{states},
    \item $\Gamma = \{0,1,z_0\}$ is the finite \textit{stack alphabet},
    \item $\delta: Q \times (\bin \cup \{\lambda\}) \times \Gamma \rightarrow Q \times \Gamma^*$ is the \emph{transition function},
    \item $\nu: Q \times (\bin \cup \{\lambda\}) \times \Gamma \rightarrow \fbins$ is the \emph{output function},
    \item $q_0 \in Q$ is the \textit{initial state},
    \item $z_0 \in \Gamma$ is the special \textit{bottom of stack symbol},
    \item $c\in \N$ is an upper bound on the number of $\lambda$-transitions per input bit.
\end{enumerate}
\label{PD: Def: PD}
\end{definition}

We write $\delta_Q$ and $\delta_{\Gamma*}$ to represent the projections of the function $\delta.$ For $z \in \Gamma^+$ the stack of $C$, $z$ is ordered such that $z[0]$ is the topmost symbol of the stack and $z[|z|-1] = z_0.$ $\delta$ is restricted to prevent $z_0$ being popped from the bottom of the stack. That is, for every $q \in Q, \, b \in \bin \cup \{\lambda\}$, either $\delta(q,b,z_0) = \bot$, or $\delta(q,b,z_0) = (q',vz_0)$ where $q' \in Q$ and $v \in \Gamma^*$.

Note that $\delta$ accepts $\lambda$ as a valid input symbol. This means that $C$ has the option to pop the top symbol from its stack and move to another state without reading an input bit. This type of transition is call a $\lambda$\textit{-transition}. In this scenario $\delta(q,\lambda, a) = (q',\lambda)$. To enforce determinism, we ensure that one of the following hold for all $q \in Q$ and $a \in \Gamma$: \begin{itemize}
    \item $\delta(q,\lambda,a) = \bot$, or
    \item $\delta(q,b,a) = \bot$ for all $b \in \bin$.
\end{itemize}
This means that the compressor does not have a choice to read either $0$ or $1$ characters. To prevent an arbitrary number of $\lambda$-transitions occurring at any one time, we restrict $\delta$ such that at most $c$ $\lambda$-transitions can be performed in succession without reading an input bit.

The transition function is extended to $\delta': Q \times \{0,1,\lambda\} \times \Gamma^+ \to Q \times \Gamma^*$ and is defined recursively as follows. For $q \in Q,\, v \in \Gamma^*,\, a \in \Gamma$ and $b \in \{0,1,\lambda\}$
\begin{align*}
    \delta'(q, b,av) = \begin{cases}
    (\delta_Q(q,b,a), \delta_{\Gamma^*}(q, b, a)v), \, &\text{if } \delta(q,b, a) \neq \bot; \\
    \bot & \text{otherwise,} 
    \end{cases}
\end{align*}
\noindent For readability, we abuse notation and write $\delta$ instead of $\delta'$. The transition function is extended further to $\delta'': Q \times \fbins \times \Gamma^{+} \to Q \times \Gamma^*$ as follows. For $q \in Q,\, v \in \Gamma^+,\, w \in \fbins$ and $b \in \{0,1.\lambda\}$

\begin{align*}
    \delta''(q, \lambda,v) = \begin{cases}
    \delta''(\delta_Q(q,\lambda,v), \lambda, \delta_{\Gamma^*}(q, \lambda, v)), \, &\text{if } \delta(q,\lambda, v) \neq \bot; \\
    (q,v) & \text{otherwise,} 
    \end{cases}
\end{align*}

\begin{align*}
    \delta''(q,wb,v) = \begin{cases}
    \delta''(\delta_Q(\delta_Q''(q,w,v),b,\delta_{\Gamma^*}''(q,w,v)),\lambda, \delta_{\Gamma^*}(\delta_Q''(q,w,v),b,\delta_{\Gamma^*}''(q,w,v)),\\   \hphantom{\bot,,} \text{if } \delta''(q,w,v) \neq \bot \text{ and } \delta(\delta_Q''(q,w,v),b,\delta_{\Gamma^*}''(q,w,v)) \neq \bot; \\
    \bot,  \text{ otherwise.}
    \end{cases}
\end{align*}

\noindent In an abuse of notation we write $\delta$ for $\delta''$ and $\delta(w)$ for $\delta(q_0,w,z_0)$.

We define the \textit{extended output function} for $q \in Q$, $b \in \bin$, $w \in \fbins$ with stack contents $a \in \Gamma, \,v \in \Gamma^*$ by the recursion $\nu'(q,\lambda, av) = \nu(q,\lambda,a) = \lambda$, $\nu'(q,b, av) = \nu(q,b,a)$ and $$\nu'(q,wb,av) = \nu'(q,w,av)\cdot\nu'(\delta_Q(q,w,av),b,\delta_{\Gamma^*}(q,w,av)).$$

\noindent In an abuse of notation we write $\nu$ for $\nu'$. The \textit{output} of $C$ on input $w \in \fbins$ is denoted by the string $C(w) = \nu(q_0,w,z_0).$

To make our notion of depth meaningful, we examine the class of \textit{information lossless pushdown compressors}.

\begin{definition}
A PDC $C$ is \textit{information lossless (IL)} if for all $x \in \fbins$, the function $x \mapsto (C(x),\delta_Q(x))$ is injective.
\label{PD: Def: ILPDC}
\end{definition}

In other words, a PDC $C$ is IL if the output and final state of $C$ on input $x$ uniquely identify $x$. We call a PDC that is IL an ILPDC. We write (IL)PDC to denote the set of all (IL)PDCs. By the identity PDC, we mean the ILPDC $I_{\textrm{PD}}$ where on every input $x$, $I_{\textrm{PD}}(x) = x.$

As part of our definition of pushdown depth, we examine ILPDCs whose stack is limited to only containing the symbol $0$ also.

\subsubsection{Unary-stack Pushdown Compressors}

UPDCs are similar to counter compressors as seen in \cite{DBLP:journals/jcss/BecherCH15}. The difference here is that for a UPDC, only a single $0$ can be popped from the stack during a single transtion, while for a counter transducer, an arbitrary number of $0$s can be popped from its stack on a single transition, i.e. its counter can be deducted by an arbitrary amount. However, the UPDC has the ability to pop off $0$s from its stack without reading a symbol via $\lambda$-transitions while the counter compressor cannot. Thus, if a counter compressor decrements its counter by the value of $k$ on a single transition, a UPDC can do the same by performing $k-1$ $\lambda$-transitions in a row before reading performing the transition of the counter compressor and popping off the final $0$.

\begin{definition}
A \emph{unary-stack pushdown compressor (UPDC)} is a $7$-tuple $$C = (Q,\Gamma, \delta, \nu, q_0, z_0, c)$$ where $Q,\delta, \nu, q_0, z_0$ and $c$ are all defined the same as for a PDC in Defintion \ref{PD: Def: PD}, while the \emph{stack alphabet} $\Gamma$ is the set $\{0,z_0\}$.
\label{UPDC def}
\end{definition}

\begin{definition}
A UPDC $C$ is \textit{information lossless (IL)} if for all $x \in \fbins$, the function $x \mapsto (C(x),\delta_Q(x))$ is injective. A UPDC which is IL is referred to as an ILUPDC.
\label{PD: Def: ILUPDC}
\end{definition}

We make the following observation regarding ILUPDCs. Let $C \in \ILUPDC$ and suppose it has been given the input $yx$. After reading the prefix $y$, if $C$'s stack height is large enough such that it never empties on reading the suffix $x$, the actual height of the stack doesn't matter. That is, any reading of $x$ with an arbitrarily large stack which is far enough away from being empty will all have a similar behaviour if starting in the same state. This is because if the stack does not empty, it has little impact on the processing of $x$. We describe this below.

\begin{remark}
Let $C \in \ILUPDC$ and suppose $C$ can perform at most $c$ $\lambda$-transitions in a row. Consider running $C$ on an input of the form $yx$ and let $q$ be the state $C$ ends in after reading $y$. If $C$'s stack has a height above $(c+1)|x|$ after reading $y$, then $C$'s stack can never be fully emptied upon reading $x$. Hence, for $k,k' \geq (c+1)|x|$ with $k\neq k'$ then $C(q,x,0^kz_0) = C(q,x,0^{k'}z_0),$ i.e. $C$ will output the same string regardless of whether the height is $k$ or $k'$. Thus, prior to reading $x$, only knowing whether the stack's height is below $(c+1)|x|$ will have any importance.
\label{PD: Rmk: height discussion}
\end{remark}

\subsection{Lempel-Ziv 78}

The Lempel-Ziv 78 algorithm (denoted LZ) \cite{DBLP:LZ78} is a lossless dictionary based compression algorithm. Given an input $x \in \fbins$, LZ parses $x$ into phrases $x = x_1x_2\ldots x_n$ such that each phrase $x_i$ is unique in the parsing, except for possibly the last phrase. Furthermore, for each phrase $x_i$, every prefix of $x_i$ also appears as a phrase in the parsing. That is, if $y \sqsubset x_i$, then $y = x_j$ for some $j<i$. Each phrase is stored in LZ's dictionary. LZ encodes $x$ by encoding each phrase as a pointer to its dictionary containing the longest proper prefix of the phrase along with the final bit of the phrase. Specifically for each phrase $x_i$, $x_i = x_{l(i)}b_i$ for $l(i) < i$ and $b_i \in \{0,1\}.$ Then for $x = x_1x_2\ldots x_n$
$$LZ(x) = c_{l(1)}b_1c_{l(2)}b_2\ldots c_{l(n)}b_n$$
where $c_i$ is a prefix free encoding of the pointer to the $i^{th}$ element of LZ's dictionary, and $x_0 = \lambda$.

For strings $w=xy$, we let $\LZ(y|x)$ denote the output of LZ on $y$ after it has already parsed $x$. For strings of the form $w = xy^n$, we use Lemma $1$ from \cite{DBLP:journals/mst/MayordomoMP11} to get an upper bound for $|\LZ(y^n|x)|$.

\begin{lemma}[\cite{DBLP:journals/mst/MayordomoMP11}]
Let $n \in \N$, and $x,y \in \fbins$ where $y \neq \lambda$. Let $w = xy^n$. Suppose on its computation of the string $w$ that $\LZ$'s dictionary contained $d \geq 0$ phrases after reading $x$. Then we have that $$|\LZ(y^n|x)| \leq \sqrt{2(|y|+1)|y^n|}\log{(d + \sqrt{2(|y|+1)|y^n|})}.$$
\label{LZ: Lemma: Repeat lemma}
\end{lemma}

\section{Finite-State Depth}

An infinitely often (i.o.) finite-state depth notion was introduced by Doty and Moser in \cite{DotyM07} based on finite-state transducers. In this section, we state and prove properties of finite-state depth found in \cite{DBLP:conf/sofsem/JordonM20} whose proofs were omitted in their original publication for space. These properties are needed to compare it with pushdown depth and LZ-depth introduced in later sections. Henceforth when we say a sequence is finite-state deep, we assume it is the i.o. version.

\subsection{Binary Representation of FSTs}

Before we begin examining depth, we first  choose a binary representation of all finite-state transducers.

\begin{definition}
A \emph{binary representation of finite-state transducers} $\sigma$ is a partially computable map $\sigma : D \subseteq \fbins \rightarrow \mathrm{FST}$, such that for every $\mathrm{FST}$ $T$, there exists some $x \in D$ such that $\sigma(x)$ fully describes $T$, i.e. $\sigma$ is surjective. If $\sigma(x) = T$, we call $x$ a $\sigma$ \emph{description} of $T$. If $x \notin D$, then $\sigma(x) = \bot.$ 
\label{FS: Def: Binary Description}
\end{definition}

For a binary representation of FSTs $\sigma$, we define 
\begin{equation}
|T|_{\sigma} = \min \{|x| : \sigma(x) = T\} \label{FS: EQ: FST Size}
\end{equation}
to be the size of $T$ with respect to $\sigma$. For all $k \in \N$, we define

\begin{equation}\FSTsize{k}_{\sigma} = \{ T \in \mathrm{FST}:|T|_{\sigma} \leq k \} \label{FS: EQ: FST k}
\end{equation}
to be the set of FSTs of with a $\sigma$ description of length $k$ or less. 
\begin{definition}
For $k \in \N$ and $x \in \fbins$, the \emph{k-finite-state complexity} of $x$ with respect to binary representation $\sigma$ is defined as $$D^{k}_{\sigma}(x) = \min \Big \{ |y| \, : \,  (T \in \mathrm{FST}^{\leq k}_{\sigma})\, T(y) = x \, \Big \}.$$
\end{definition} Here $y$ is the shortest string that gives $x$ as an output when inputted into an FST of size $k$ or less with respect to the binary representation $\sigma$. If no such $y$ exists we say $\FScomp{k}(x) = \infty.$ $T$ can be thought of as the FST that can decompress $y$ to reproduce $x$.

For the purpose of this paper, we fix a binary representation of finite-state transducers $\sigma$. Let $T = (Q,q_0,\delta,\nu)$ be an FST. We define the function $\Delta: Q \times \bin \rightarrow Q \times \fbins$, where \begin{equation}\Delta(q,b) = (\delta(q,b),\nu(q,b))
\label{FS: EQ: Delta}
\end{equation} which completely describes the state transitions and outputs of $T$. Calude, Salomaa and Roblot previously presented different methods to encode $\Delta$ in \cite{calude2011finite}. Further study of binary representations of FSTs can be found in \cite{calude2016finite}. We re-present the first encoding scheme from \cite{calude2011finite} here (borrowing the notation they use) as it is used to define our own binary representation later.

For $n \in \N$, let $\binf(n)$ denote the binary representation of $n$. For instance $\binf(1) = 1, \, \binf(2) = 10,\,  \binf(3) = 11$. Note that $\binf(n)$ begins with a $1$ for all $n$. $\stringf(n)$ denotes the binary string built by removing the first $1$ in $\binf(n)$. So, $\binf(n) = 1\cdot \stringf(n)$. Note that $|\stringf(n)| = \lfloor \log (n) \rfloor.$

For $x = x_1x_2\ldots x_l$, where $x_i \in \bin$ for $1 \leq i \leq l$, we define the following two strings:

\begin{enumerate}
    \item $x^{\dagger} = x_10x_20\ldots x_{l-1}0x_l1$, and
    \item $x^{\diamond} = \overline{(1x)^{\dagger}},$
\end{enumerate}
where $\bar{0} = 1$ and $\bar{1} = 0.$

Then if $Q = \{q_1,\ldots ,q_m\}$, $\Delta$ is encoded by the binary string

\begin{equation}\pi = \binf(n_1)^{\ddagger}\cdot \stringf(n_1')^{\diamond}\cdot \binf(n_2)^{\ddagger} \cdots \binf(n_{2m})^{\ddagger}\cdot \stringf(n_{2m}')^{\diamond}, \label{FS: EQ: Calude Binary}
\end{equation}
where $\Delta(q_i,b) = (q_{(1+ n_{2i - 1 + b}\mod m)}, \stringf(n_{2i-1+b}')),  1 \leq i \leq m,$ and $b \in \bin$. Here, $\binf(n_t)^{\ddagger} = \lambda$ if the corresponding transition stays in the same state, that is $\delta(q_t,b) = q_t.$ Otherwise $\binf(n_t)^{\ddagger} = \binf(n_t)^{\dagger}.$

While $\pi$ in \eqref{FS: EQ: Calude Binary} gives a complete description of $\Delta$, there is no indication of what the initial state of $T$ is. One is left to assume that $q_1$ is the initial state. This leads to the question of whether changing the initial state of $T$ to $q_i \neq q_1$ drastically alters the length of the corresponding encoding of the new FST.

To overcome this, the binary representation $\sigma: D \rightarrow \text{FST}$ for FSTs we use is as follows. Let $$\Delta_m = \{ \pi \, | \pi \, \text{is an encoding of $\Delta$ for an FST with $m$ states.}\}$$ be the set of all possible encodings of $\Delta$ for all FSTs with $m$ states. The domain $D$ of $\sigma$ is the set of strings $$ D = \bigcup\limits_{m} \bigcup\limits_{1 \leq i \leq m} \{d(\binf(i) ) 01 y \, | \, y \in \Delta_m \}.$$ Then for $1\leq i \leq m$ and $y \in \Delta_m$ we set \begin{equation}\sigma(d(\binf(i))01y) = T \label{FS: EQ: Binary Rep} \end{equation}where $T$ is the FST with $Q = \{q_1, \ldots q_m\}$ with initial state $q_0 = q_i$ and whose transition function $\Delta$ is described by $y$. Clearly $\sigma$ is surjective and so is a binary representation of all FSTs.

We require a pointer to the initial state as for two transducers which are equivalent up to a relabelling of their states, this change of relabelling of states changes the encoding of their respective $\Delta$. This pointer allows us to easily get a bound on the size of transducers with equivalent transition tables, but different initial states. Specifically our binary representation $\sigma$ is used to prove Lemma \ref{FS: Lemma: GEQ Lemma}. However, Lemma \ref{FS: Lemma: Deep for every binary rep} demonstrates that if a sequence is deep when the size of transducers is viewed from the perspective of one binary representation, it is deep when viewed from the perspective of any binary representation. Henceforth, we will drop the $\sigma$ notation and instead write $|T|$ for $|T|_{\sigma}$, FST$^{\leq k}$ for FST$_{\sigma}^{\leq k}$ and $\FScomp{k}(x)$ instead of $D^{k}_{\sigma}(x)$. All other definitions and results hold and can be proved regardless of the binary representation being used.

To measure the randomness of a sequence in the finite-state setting, we use the following notions of dimension.

\begin{definition}
Let $S\in \infbins$.
\begin{enumerate}
    \item The \emph{finite-state dimension} of $S$ \cite{DBLP:journals/tcs/DaiLLM04} is defined to be \begin{equation}\text{dim}_{\text{FS}}(S) = \lim_{k\to\infty}\liminf\limits_{n\rightarrow \infty}\frac{\FScomp{k}(S \upharpoonright n)}{n} = \inf_{T \in \ILFST}\liminf_{n \rightarrow \infty}\frac{|T(S \upharpoonright n)|}{n}.\end{equation}
    \item The \emph{finite-state strong dimension} of $S$ \cite{DBLP:journals/siamcomp/AthreyaHLM07} is defined to be \begin{equation}\text{Dim}_{\text{FS}}(S) = \lim_{k\to\infty}\limsup\limits_{n\rightarrow \infty}\frac{\FScomp{k}(S \upharpoonright n)}{n} = \inf_{T \in \ILFST}\limsup_{n \rightarrow \infty}\frac{|T(S \upharpoonright n)|}{n} .\end{equation}
\end{enumerate}
\label{FS: Def: Dimension}
\end{definition}

The original definitions of finite-state and finite-state strong dimension presented in \cite{DBLP:journals/siamcomp/AthreyaHLM07,DBLP:journals/tcs/DaiLLM04} were based on finite-state gamblers and information lossless finite-state compressors. However, using the relationship between finite-state compressors and decompressors as shown in  \cite{DotyMoserLossy,DBLP:journals/iandc/SheinwaldLZ95}, the equalities above relating dimension to $k$-finite-state complexity hold also.
 
 Note that $\text{dim}_{\text{FS}}(S) = \rho_{\ILFST}(S)$ and $\text{Dim}_{\text{FS}}(S) = R_{\ILFST}(S)$.

\subsection{Finite-State Depth}

A sequence $S$ is \textit{finite-state deep} if, given any finite-state transducer, we can always build a more powerful finite-state transducer (i.e. via a combination of having more states to process the input and the ability to output longer strings) such that when we examine the $k$-finite-state complexity of prefixes of $S$ on each transducer, their difference is always bounded below by the length of the prefix times a fixed constant. Intuitively, the larger transducer is more powerful and can spot patterns of the sequence that the smaller transducer cannot not. As such, the larger transducer requires less bits to describe the prefix. In \cite{DotyM07}, a notion of depth based on FSTs was introduced which we give the definition of below.

\begin{definition}
A sequence $S$ is \textit{infinitely often finite-state deep (FS-deep)} if 
$$(\exists \alpha > 0)(\forall k \in \N)(\exists k' \in \N)(\exists^{\infty} n \in \N) \, \FScomp{k}(S \upharpoonright n) - \FScomp{k'}(S \upharpoonright n) \geq \alpha n.$$ \label{FS: Def: io fs deep}
\end{definition}

The following lemma demonstrates that if a sequence $S$ is FS-deep when the size of finite-state transducers are viewed with respect to one binary representation, then it is FS-deep regardless of what binary representation is used.

\begin{lemma}
Let $\pi$ be a binary representation of FSTs. Let $S$ be a FS-deep sequence when the size of the FSTs are viewed with respect to the binary representation $\pi$. Then $S$ is FS-deep when the size of the FSTs are viewed with respect to every binary representation.
\label{FS: Lemma: Deep for every binary rep}
\end{lemma}

\begin{proof}
Let $S$ and $\pi$ be as in the statement of the lemma. Let $\tau$ be a different binary representation of all FSTs. Fix $k\in\N$. Then there exists a constant $c$ such that $\mathrm{FST}_{\tau}^{\leq k} \subseteq \mathrm{FST}_{\pi}^{\leq k+c}$. Therefore for all $n\in\N,$  
\begin{equation}D^{k+c}_{\pi}(S \upharpoonright n) \leq D^{k}_{\tau}(S \upharpoonright n).\label{FS: Lemma: Deep for every binary rep: EQ1}\end{equation}

As $S$ is FS-deep with respect to $\pi$, there exists constants $\alpha$ and $(k+c)'$ such that for infinitely many $n$ \begin{equation}D^{k+c}_{\pi}(S\upharpoonright n) - D^{(k+c)'}_{\pi}(S\upharpoonright n) \geq \alpha n. \label{FS: Lemma: Deep for every binary rep: EQ2}\end{equation}

Let $d$ be a constant such that $\mathrm{FST}_{\pi}^{(k+c)'} \subseteq \mathrm{FST}_{\tau}^{(k+c)' + d}.$ Therefore for almost every $n$,
\begin{equation}
    D_{\tau}^{(k+c)'+d}(S\upharpoonright n) \leq D^{(k+c)'}_{\pi}(S\upharpoonright n). \label{FS: Lemma: Deep for every binary rep: EQ3}
\end{equation}

Therefore for infintely many $n$,
\begin{equation}
D^{k}_{\tau}(S\upharpoonright n) - D_{\tau}^{k'}(S\upharpoonright n)
\geq D^{k+c}_{\pi}(S\upharpoonright n) - D^{(k+c)'}_{\pi}(S\upharpoonright n) \geq \alpha n.
\end{equation}
where $k' = (k+c)'+d$. As $k$ is arbitrary, we have that $S$ is also FS-deep with respect to $\tau$.

\end{proof}

The following two lemmas demonstrate the relationship between the $k$-finite-state complexity of strings $x$ and $M(x)$ where $M$ is an ILFST. It previously appeared in \cite{DotyM07} but we restate and reprove parts of it here.

\begin{lemma}[\cite{DotyM07}]
Let $M$ be an $\mathrm{ILFST}$. 
\begin{enumerate}
    \item $(\forall k \in \N)(\exists k' \in \N)(\forall x \in \fbins) \, \FScomp{k'}(M(x)) \leq \FScomp{k}(x).$
    
    \item $(\forall \varepsilon >0)(\forall k \in \N)(\exists k' \in \N)(\forall^{\infty} x \in \fbins) \, \FScomp{k'}(x) \leq (1 + \varepsilon)\FScomp{k}(M(x)) + O(1).$
\end{enumerate}
\label{FS: Lemma: ILFST Machine less/greater}
\end{lemma}

\begin{proof}
The proof for part $1$ is in \cite{DotyM07}.

For part $2$, let $\varepsilon,k, x$ and $M$ be as stated in the lemma. Furthermore let $0<\varepsilon' < \varepsilon.$ By Theorem \ref{FS: Thm: inverse fst}, there exists an ILFST $M^{-1}$ and a constant $c\in \N$ such that for all $y \in \fbins$, $y\upharpoonright(|y|-c) \sqsubseteq M^{-1}(M(y))\sqsubseteq y.$ 

Let $p$ be a $k$-minimal program for $M(x)$, i.e. $A(p) = M(x)$ for $A \in \FSTsize{k}$ and $\FScomp{k}(M(x)) = |p|$.

Let $b = \lceil \frac{2}{\varepsilon'}\rceil$. There exits non-negative integers $n$ and $r$ such that $|p| = nb + r$, where $0\leq r< b$. Let $p'$ be a new string such that $p'$ begins with the first $nb$ bits of $p$, with a $0$ placed to separate every $b$ bits starting at the beginning of the string, followed by a $1$ and then the remaining $r$ bits of $p$ doubled. That is $$p' = 0p_1\ldots p_b0p_{b+1}\ldots p_{2b}0\ldots p_{nb}1p_{nb+1}p_{nb+1}\ldots p_{nb + r}p_{nb + r}.$$ Note therefore that \begin{equation}
    |p'| = n(b+1) + 2r + 1 = |p| + n + r + 1 < |p| + n + b + 1.\label{FS: Lemma: ILFST Machine less/greater: EQ1}
\end{equation}

\noindent
$nb \leq |p|$ means $n \leq \big\lceil \frac{|p|}{b}\big\rceil$ and so for $|p|$ large it holds that
\begin{align}
    |p'| &\leq |p| + \bigg\lceil \frac{|p|}{b}\bigg\rceil + b + 1 \leq |p| + 2\bigg\lceil \frac{|p|}{b}\bigg\rceil 
     = |p| + 2\bigg\lceil \frac{|p|}{\lceil \frac{2}{\varepsilon'}\rceil}\bigg\rceil \notag \\
     &\leq |p| + 2(\frac{\varepsilon' |p|}{2} + 1) 
     = |p|(1+\varepsilon') + 2 \notag \\
     &\leq |p|(1 + \varepsilon). \label{FS: Lemma: ILFST Machine less/greater: EQ2}
\end{align}

Next we build $A'$ for $x$. Let $y = M^{-1}M(x))$, i.e. $x = yz$ for some $|z| \leq c$. Let $A'$ be the machine such that on input $p'01z$: $A'$ uses $p'$ to simulate $A(p)$ to retrieve $M(x)$. $A'$ knows where $p'$ ends due to the $01$ separator. $A'$ takes $M(x)$'s output and simulates it on $M^{-1}$ to retrieve $y$. This is possible as the composition of ILFSTs can be computed by an ILFST. After seeing the separator, $A'$ acts as the identity transducer and outputs $z$. Thus $A'(p'01z) = yz = x$. Thus
\begin{equation}\FScomp{|A'|}(x) \leq |p'| + 2 + |z| \leq |p|(1 + \varepsilon) + 2 + c = \FScomp{k}(M(x)) + O(1). \label{FS: Lemma: ILFST Machine less/greater: EQ3}\end{equation}

As $|A'|$ depends on the size of $A$, the size of $M^{-1}$ and $b$ (i.e. it does not vary with $x$), $k'$ can be chosen such that $k' = |A'|.$

\end{proof}

Lemmas \ref{FS: Lemma: GEQ Lemma} and \ref{FS: Lemma: LEQ Lemma} below examine the $k$-finite-state complexity of substrings within a string on FSTs of roughly the same size. Lemma \ref{FS: Lemma: GEQ Lemma}'s proof relies on viewing FSTs with respect to our fixed binary representation $\sigma$. However, this has no impact on whether a sequence is deep or not by Lemma \ref{FS: Lemma: Deep for every binary rep}. Suppose we are given an FST $T$ and an input $vw$. If $T$ outputs $xy$ on reading $vw$ and $x$ on reading the prefix $v$, this means that $v$ is a description of $x$ and $vw$ is a description of $xy$ via $T$. We can alter $T$ to create a new transducer $T'$ such that the states and transitions of $T'$ and $T$ are the same, with the only difference being that the start state of $T'$ is the one $T$ ends in after reading $v$. This means that $w$ is a description of $y$ via $T'$.

\begin{lemma}
For our fixed binary representation $\sigma$ (from \ref{FS: EQ: Binary Rep}),
$$(k \geq 4)(\forall n \in \N)(\forall x,y,z \in \fbins) \, \FScomp{k}(xy^nz) \geq \FScomp{3k}(x) + n\FScomp{3k}(y) + \FScomp{3k}(z).$$
\label{FS: Lemma: GEQ Lemma}
\end{lemma}

\begin{proof}
Let $k,n,x,y,z$ be as in the lemma. Let $T \in \FSTsize{k}$ and $p_x,p_{y,i},\ldots p_{y,n},p_z \in \fbins$ be such that $\FScomp{k}(xy^nz) = |p_xp_{y,1}\ldots p_{y,n}p_z|$ with $T(p_xp_{y,1}\ldots p_{y,n}p_z) = xy^nz,$ $T(p_xp_{y,1}\ldots p_{y,j}) = xy^j$ for $1 \leq j \leq n$, and $T(p_x) = x$.

For all $w\in \fbins$, let $T_w$ be the FST such that $T_w$'s states, transitions and outputs are the same as $T$'s with the only difference being that the start state of $T_w$ is the state that $T$ on input $w$ ends in.
So $T_{p_xp_{y,1} \ldots p_{y,j-1}}(p_{y,j}) = y$ and $T_{p_xp_{y,1} \ldots p_{y,n}}(p_z) = z$.

Next we put a bound on the binary description length of $T_w$. Recall that for $\sigma$, for an FST $M$, $\sigma(d(bin(n))01\pi) = M$ where $d(bin(n))$ is a pointer to $M$'s start state $q_n$, and $\pi$ describes the function $\Delta$ of $M$'s transitions and outputs. We write  $\Delta_M$ for the $\Delta$ function of $M$.

As $|T| \leq k$, $T$ has at most $k$ states. Therefore the pointer to $T_w$'s start state takes at most $2|\binf(k)| = 2(\lfloor \log k \rfloor + 1)$ bits to encode. Similarly, the encoding of $\Delta_{T}$ can be used to encode $\Delta_{T_w}$ and so the number of bits required to encode $\Delta_{T_w}$ is bounded above by $k$ bits also. Hence we have that whenever $k \geq 4$ \begin{equation}
    |T_w| \leq 2(\lfloor \log k \rfloor + 1) + 2 + k \leq 3k. \label{FS: Lemma: GEQ Lemma: EQ1}
\end{equation} 

\noindent Thus for $k \geq 4$, we have that $\FScomp{3k}(x) \leq |p_x|,$ $\FScomp{3k}(z)\leq |p_z|$ and $\FScomp{3k}(y) \leq |p_y|$ where $|p_y| = \min\big\{ |p_{y,j}| : 1 \leq j \leq n \big \}.$

Hence we have that
\begin{align*}
    \FScomp{k}(xy^nz)  = |p_xp_{y,1}\ldots p_{y,n}p_z| 
     \geq |p_x| + n|p_y| + |p_z| 
     \geq \FScomp{3k}(x) + n \FScomp{3k}(y) + \FScomp{3k}(z)
\end{align*}
as desired.

\end{proof}

\begin{remark} Note that for all $x \in \fbins,$ $\FScomp{k}(x) \leq \FScomp{k+1}(x)$. Hence while Lemma \ref{FS: Lemma: GEQ Lemma}'s result is only for $k \geq 4$, it gives us that for all $x,y,z$, $$\FScomp{1}(xy^nz) \geq \FScomp{2}(xy^nz) \geq  \FScomp{3}(xy^nz) \geq \FScomp{12}(x) + n \FScomp{12}(y) + \FScomp{12}(z).$$
\label{FS: Remark: k less than 4}
\end{remark}

\begin{remark}
Lemma \ref{FS: Lemma: GEQ Lemma} can be generalised such that for our fixed binary representation $\sigma$, we can break the input into any number of substrings to get a similar result. That is for any string $x = x_1\ldots x_n$,  $$(\forall^{\infty} k \in \N) \, \FScomp{k}(x_1 \ldots x_n) \geq \sum_{i=1}^n\FScomp{3k}(x_i).$$
\label{FS: Remark: any substring split}
\end{remark}

The following lemma states that for almost every pair of strings $x$ and $y$, given a description of $x$ and a description of $y$, a transducer $T$ can be built such that upon reading a padded version of the description of $x$, a flag, and the description for $y$, $T$ can output the string $xy$.

\begin{lemma}

$(\forall \varepsilon > 0)( \forall k \in \N)( \exists k' \in \N) (\forall^{\infty} x \in \fbins)(\forall y \in \fbins)$\,
$$\FScomp{k'}(xy) \leq (1+\varepsilon)\FScomp{k}(x) + \FScomp{k}(y) + 2.$$
\label{FS: Lemma: LEQ Lemma}
\end{lemma}

\begin{proof}
Let $\varepsilon, x,y$ and $k$ be as stated in the lemma. Consider $p,q \in \fbins$ such that $\FScomp{k}(x) = |p|$ and $\FScomp{k}(y) = |q|$, and suppose $A,B \in \FSTsize{k}$ where $A(p) = x$ and $B(q) = y$.

Let $b = \lceil \frac{2}{\varepsilon}\rceil$. Then there exists integers $n$ an $r$ such that $|p| = nb + r$, where $0\leq r< b$. Let $p'$ be a new string such that $p'$ begins with the first $nb$ bits of $p$, with a $0$ placed to separate every $b$ bits starting at the beginning of the string. This is followed by a $1$ and the remaining $r$ bits of $p$ doubled, i.e. $$p' = 0p_1\ldots p_b0p_{b+1}\ldots p_{2b}0\ldots p_{nb}1p_{nb+1}p_{nb+1}\ldots p_{nb + r}p_{nb + r}.$$ Then by the same argument as in Lemma \ref{FS: Lemma: GEQ Lemma}, whenever $|p|$ is large enough we can arrive to the same result as in Equation \eqref{FS: Lemma: ILFST Machine less/greater: EQ2}, i.e. it holds that $|p'| \leq |p|(1 + \varepsilon)$. Another way of saying this holds only for large $|p|$ is when $\FScomp{k}(x)$ is large. Hence the $(\forall^{\infty}x \in \fbins)$ requirement in the statement of the lemma.

Next let $T \in \FSTsize{k'}$ where $k'$ is a number whose value is dependent only on $k$ and $b$ be the following FST such that on input $p'10q$: $M$ uses $p'$ to output $A(p)=x$. $M$ can spot the beginning bits of $p$ from the blocks of size $b$ by the $0$s. When $M$ sees the block beginning with $1$ it knows that the remaining bits will be the final bits of $p$ doubled. Upon reading $10$, $M$ uses the remaining bits to output $B(q)=y$. Therefore, for almost all $x$ and all $y$ it holds that
\begin{equation*}
    \FScomp{k'}(xy) \leq |p'| + |q| + 2 \leq (1+\varepsilon)\FScomp{k}(x) + \FScomp{k}(y) + 2.
\end{equation*}

\end{proof}

\section{Pushdown Depth}

This section presents our new notion of pushdown depth (PD-depth) based on pushdown compressors that differs from the notion in \cite{DBLP:conf/sofsem/JordonM20} and compares it with FS-depth. It is an almost everywhere notion. Our definition examines the difference between the performance between all ILUPDCs and an ILPDC $C'$. Intuitively, a sequence $S$ is pushdown deep if $S$ contains some structure which ILUPDCs cannot exploit during compresisson due to their stack restriction while $C'$ can.

\begin{definition}
Let $S \in \infbins$. $S$ is \emph{pushdown deep (PD-deep)} if  $$(\exists \alpha > 0)(\forall C \in \ILUPDC)(\exists C' \in \ILPDC)(\forall^{\infty} n \in \N),\,  |C(S \upharpoonright n)| - |C'(S \upharpoonright n)| \geq \alpha n.$$ \label{PD: Def: PD deep}
\end{definition}

\subsection{Basic Properties of Pushdown Depth}

The following theorem shows that PD-depth satisfies two of the fundamental depth properties in that both ILUPDC-trivial sequences (in the sense that $R_{\textrm{UPD}}(S) = 0$) and ILPDC-incompressible sequences (in the sense that $\rho_{\textrm{PD}}(S) = 1$) are not PD-deep. This is analogous to computable and Martin-L\"{o}f-random sequences being shallow in Bennett's original depth notion. 

\begin{theorem}
Let $S \in \infbins$.
\begin{enumerate}
    \item If $\rho_{\mathrm{PD}}(S) = 1$, then $S$ is not $\mathrm{PD}$-deep.
    \item If $R_{\mathrm{UPD}}(S) = 0$, then $S$ is not $\mathrm{PD}$-deep.
\end{enumerate}
\label{PD: Thm: Easy Hard}
\end{theorem}

\begin{proof}

Let $S \in \infbins$ be such that $\rho_{\mathrm{PD}}(S) = 1$. Therefore for every $\alpha > 0$ and every $C \in \mathrm{ILPDC}$, for almost every $n$ \begin{equation}|C(S \upharpoonright n)| > n(1 - \alpha). \label{PD: Thm: Easy Hard: EQ1}\end{equation}

\noindent Then for almost every $n$
\begin{equation}|I_{\mathrm{PD}}(S \upharpoonright n)| - |C(S \upharpoonright n)| < n - n(1-\alpha) = \alpha n.\label{PD: Thm: Easy Hard: EQ2}\end{equation} As $\alpha$ is arbitrary and $I_{\mathrm{PD}} \in \ILUPDC$, $S$ is therefore not PD-deep.

Next suppose $S \in \infbins$ is such that $R_{\mathrm{UPD}}(S) = 0$. Let $C \in \ILUPDC$ be such that $\limsup\limits_{n \to \infty} 
|C(S \upharpoonright n)|/n = 0.$ Hence for every $\beta > 0$ and almost every $n$, \begin{equation}|C(S \upharpoonright n)| < \beta n.\label{PD: Thm: Easy Hard: EQ3}\end{equation}  

\noindent Therefore for every $C' \in  \ILPDC$, it holds that for almost every $n$ \begin{equation}
    |C(S \upharpoonright n)| - |C'(S \upharpoonright n)| \leq |C(S \upharpoonright n)| < \beta n.\label{PD: Thm: Easy Hard: EQ4}\end{equation}
As $\beta$ is arbitrary, $S$ is not PD-deep.

\end{proof}

Before we prove a slow growth law for pushdown depth, we first demonstrate that the composition of any ILPDC (or ILUPDC) $C$ with any ILFST $T$ can be simulated by another ILPDC (or ILUPDC) $N$ which is allowed to perform more $\lambda$-transitions than $C$.

\begin{lemma} 
Given $C \in \ILPDC$ (similarly $C \in \ILUPDC$) and $T \in  \ILFST$, we can build an $\ILPDC$ (similarly an $\ILUPDC$) $N$, such that $\forall x \in \fbins$, $N(x) = C(T(x))$.
\label{PD: Lemma: PD FST Composition}
\end{lemma}

\begin{proof}

Let $T = (Q_T,q_{0,T},\delta_T,\nu_T)$ and $C = (Q_C, \Gamma_C,\delta_C,\nu_C,q_{0,C},z_0,c)$ be an ILFST and an ILPDC respectively as stated in the lemma. Let $d = \max \{ |T(q,b)| : q \in Q_T, b \in \bin \}$ denote the longest output possible from a transition in $T$. We build the PDC $N = (Q_N,\Gamma_C,\delta_N,\nu_N,q_{0,N},z_0,cd),$ where
\begin{itemize}
    \item $Q_N = Q_C \times Q_T \times S,$ where $S = \bin^{\leq cd}$,
    \item $q_{0,N} = (q_{0,C},q_{0,T},\lambda).$
\end{itemize}

$N$ works as follows: Before reading a bit, $N$ uses $\lambda$-transitions to pop the topmost $cd$ bits of its stack, or until the stack only contains $z_0$, and remembers them in its states. That is, while $|y| < cd$ and $a \neq z_0,$ $$\delta_N((q_C,q_T,y),\lambda,a) = ((q_C,q_T,ya), \lambda).$$ On such states, $$\nu_N((q_C,q_T,y),\lambda,a) = \lambda.$$

Then for $b \in \bin$, if $a = z_0$ or $|y| = cd$, $N$ moves to the state representing how $C$ would move on input $\nu_T(q_T,b)$, how $T$ would move on input $b$, and to the state representing not having the topmost stack bits in memory. $N$'s stack then updates to be the same as $C$'s would be as if it had read $\nu_T(q_T,b)$. That is, $$\delta((q_C,q_T,y),b,a) = ((\delta_{C,Q}(q_C,\nu_T(q_T,b),ya),\delta_{T,Q}(q_T,b),\lambda),xa)$$ where for some $w \in \fbins$ either
\begin{enumerate}
    \item $x = wy$, if $C$ would have pushed $w$ onto its stack reading $\nu_T(q_T,b)$,
    
    \item $x = wy[i \ldots |y|-1]$, if $C$ would have popped off the top $i$ symbols and then pushed $w$ onto its stack reading $\nu_T(q_T,b)$,
    
    \item $x =y[i \ldots |y|-1]$, if $C$ would have popped off the top $i$ symbols from its stack and pushed nothing on when reading $\nu_T(q_T,b)$.
\end{enumerate} 

\noindent As there are only a finite number of possibilities, these can all be coded into the states and transitions. On such states, $$\nu_N((q_C,q_T,y),b,a) = \nu_C(q_C,\nu_T(q_T,b),ya).$$

$N$ is an ILPDC as from knowledge of the output and $q_C$, we can recover $T(x)$ as $C$ is IL, and from $q_T$ and $T(x)$ we can recover $x$ as $T$ is IL.

Note that if $C \in \ILUPDC$, the proof remains the same except $S = \{0\}^{\leq cd}$ and $w$ mentioned above is an element of $\{0\}^*$.

\end{proof}

The following result shows that pushdown depth satisfies a slow growth law.

\begin{theorem}[Slow Growth Law]
Let $S\in \infbins$, let $g:\infbins \rightarrow \infbins$ be ILFS computable and let $S' = g(S)$. If $S'$ is PD-deep then $S$ is PD-deep.
\label{PD: Thm: SGL}
\end{theorem}

\begin{proof}
Let $S, S',f$ and $g$ be as in the statement of the lemma and $T$ be the ILFST computing $g$.

For all $n\in \N$ such that $T(S \upharpoonright m) = S' \upharpoonright n$ for some $m$, let $m_n$ denote the largest integer such that $T(S \upharpoonright m_n) = S' \upharpoonright n$. Note that for all $m$, there exists an $n$ such that $M(S \upharpoonright m_{n-1}) \sqsubset M(S \upharpoonright m) \sqsubseteq M(S \upharpoonright m_n).$ As $M$ is IL, it cannot visit the same state twice without outputting at least one bit, so there exists a $\beta > 0$ such that for such $n$, $n\geq \beta m_n$. Furthermore recall from Theorem \ref{FS: Thm: inverse fst} that there exists an ILFST $T^{-1}$ and a constant $a$ such that for all $x \in \fbins$, $x \upharpoonright (|x|-a) \sqsubseteq T^{-1}(T(x)) \sqsubseteq x.$

Let $C \in \ILUPDC$. Let $N$ be the $\ILUPDC$ given by Lemma \ref{PD: Lemma: PD FST Composition} such that $N(x) = C(T^{-1}(x))$ for all $x$. Note that for some $n$, 
\begin{align}
|C(S \upharpoonright m)| &\geq |C(T^{-1}(T(S\upharpoonright m))|
= |N(T(S \upharpoonright m))| \notag\\
&= |N(T(S \upharpoonright m_n))| = |N(S' \upharpoonright n)|. \label{PD: Thm: SGL: EQ1}
\end{align}
\noindent As $S'$ is deep, there exists $\alpha > 0$ and an ILPDC $N'$ such that for almost every $m$, \begin{equation}
|N(S' \upharpoonright n)| - |N'(S' \upharpoonright n)| \geq \alpha n.\label{PD: Thm: SGL: EQ2}
\end{equation}

Next let $C'$ be the $\ILPDC$ given by Lemma \ref{PD: Lemma: PD FST Composition} such that on input $x$, $C'(x) = N'(T(x))$ for all $x$. Hence for some $n$, \begin{align}
    |C'(S \upharpoonright m)| = |N'(T(S \upharpoonright m))| = |N'(T(S \upharpoonright m_n))| = |N'(S' \upharpoonright n)|. \label{PD: Thm: SGL: EQ3}
\end{align}
Therefore for almost every $m \in \N$, there exists some $n$ such that
\begin{align}
    |C(S \upharpoonright m)| - |C'(S \upharpoonright m)| & \geq |N(S' \upharpoonright n)| - |N'(S' \upharpoonright n)| \tag{by \eqref{PD: Thm: SGL: EQ1} and \eqref{PD: Thm: SGL: EQ3}}\\
    & \geq \alpha n \tag{by \eqref{PD: Thm: SGL: EQ2}} \\
    & \geq \alpha \beta m_n \geq \alpha \beta m.
\end{align}
Hence $S$ is PD-deep.

\end{proof}

\subsection{Separation from Finite-State Depth}

Prior to comparing pushdown depth with Doty and Moser's i.o. finite-state depth, we require the following definition which defines the \textit{depth-level} of a sequence. Simply put, the depth-level of a deep sequence is the $\alpha$ term in Definitions \ref{FS: Def: io fs deep} and \ref{PD: Def: PD deep}.

\begin{definition}
Let $S \in \infbins$. Let $\alpha > 0$.
\begin{enumerate}
    \item We say that FS-\textit{depth}$(S) \geq \alpha$ if $$(\forall k \in \N)(\exists k' \in \N)(\exists^{\infty} n \in \N) \, \FScomp{k}(S \upharpoonright n) - \FScomp{k'}(S \upharpoonright n) \geq \alpha n.$$
    Otherwise we say FS-\textit{depth}$(S) < \alpha$.
    
    \item  We say that PD-\textit{depth}$(S) \geq \alpha$ if $$(\forall C \in \ILUPDC)(\exists C' \in \ILPDC)(\forall^{\infty} n \in \N) \, |C(S \upharpoonright n)| - |C'(S \upharpoonright n)| \geq \alpha n.$$
    Otherwise we say PD-\textit{depth}$(S) < \alpha$.
\end{enumerate}
\label{PD: Def: Depth Level}
\end{definition}

The following result demonstrates the existence of a sequence which has a large FS-depth level but not even a small PD-depth level. This sequence is composed of chunks of random strings which grow exponentially. Some of these chunks are composed of repetitions of random strings which small FSTs cannot identify while larger FSTs can, resulting in finite-state depth. Other chunks $x$ are such that $K(x) \geq |x|$ preventing the sequence being PD-deep. The construction takes advantage of the fact that one is an i.o. notion (Moser and Doty's notion) while the other is an a.e. notion.

\begin{theorem}
There exists a sequence $S$ such that for all $0 < \alpha < 1$, $\mathrm{FS}$-depth(S) $> (1-\alpha)$ and $\mathrm{PD}$-depth(S) $< \alpha$.
\label{PD: Thm: fs not PD}
\end{theorem}

\begin{proof}
Let $0 < \alpha < 1.$ Begin by partitioning the non-negative integers into disjoint consecutive intervals $I_1,I_2,\ldots$ where $|I_1| = 2$ and $|I_j| = 2^{|I_1| + \cdots + |I_{j-1}|}$ for $j \geq 2.$ For instance, $I_1 = \{0,1\},$ $I_2 = \{2,3,4,5\}$, and $I_3 = \{6,7,\ldots, 69\}.$ For each $I_j$, set $m_j = \min(I_j)$ and $M_j = \max(I_j).$ Note that for all $j$, $m_{j+1} = M_j + 1$. $S$ is constructed in stages $S_1S_2S_3 \ldots$, where $S_j$ denotes the substring $S[m_j..M_j]$ . For all $j$, we henceforth use the notation $\overline{S_j}$ to denotes the prefix $S_1\cdots S_j$ of $S$. Note that $|S_j| = |I_j|$ and when $j > 1$ we have that $\log(|S_j|) = |\overline{S_{j-1}}|$.

For every interval $I_j$ where $j$ is odd, we set $S_j$ to be a string with maximal plain Kolmogorov complexity in the sense that $K(S_j) \geq |S_j|$. If $j$ is even, $I_j$ is devoted to some FST description bound length $k$. Specifically for each $k$, $k$ is devoted to every interval $I_j$ where $j$ is of the form $j = 2^k + t(2^{k+1}),$ for $t \geq 0$. So $k=1$ is first devoted to $I_2$ and every $4^\thh$ interval after that and $k=2$ is first devoted to $I_4$ and every $8^\thh$ interval after that, and so on. 

For each $k$, let $r_k$ be a string of length $|I_{2^k}|$ such that $r_k$ is $3k$-FS random in the sense that \begin{equation}
    \FScomp{3k}(r_k) \geq |r_k| - 4k. \label{PD: Thm: fs not PD: EQ1}
\end{equation} Such a string exists as $|\FSTsize{3k}|\cdot 2^{|r_k| - 4k} < 2^{|r_k|}$ for $k \geq 1$. Then if $I_j$ is devoted to $k$ we set 
\begin{align}
    S_j = r_k^{\frac{|I_j|}{|r_k|}} = r_k^{2^{t(2^{k+1})}}. \label{PD: Thm: fs not PD: EQ2}
\end{align}

First we show $\mathrm{FS}$-depth$(S) > (1-\alpha)$. Fix some $k \geq 4$ (this allows us to apply Lemma \ref{FS: Lemma: GEQ Lemma}). We consider prefixes of the form $\overline{S_j}$ of $S$ where interval $I_j$ is devoted to $k$. We first find a lower bound for the $k$-finite-state complexity of $\overline{S_j}$. We have that \begin{align}
    \FScomp{k}(\overline{S_j}) & \geq \FScomp{3k}(\overline{S_{j-1}}) + \frac{|S_j|}{|r_k|}\FScomp{3k}(r_k) \tag{by Lemma \ref{FS: Lemma: GEQ Lemma}} \\
    & \geq \frac{|S_j|}{|r_k|}(|r_k| - 4k). \tag{by \eqref{PD: Thm: fs not PD: EQ1}}
\end{align}

For each $r \in \fbins$, consider the single state FST $T_r$ such that for each bit of its input read, $T_r$ stays in the same state and outputs $r$. That is, for all $x \in \fbins$, $T_r(x) = r^{|x|}$. Let $\widehat{k}$ be large enough so that $T_{r_k}$ and the identity transducer are contained in $\FSTsize{\widehat{k}}$. This enables us to get upper bounds for the $\widehat{k}$-finite-state complexity of $\overline{S_{j-1}}$ and $S_j$ of
\begin{align}
    \FScomp{\widehat{k}}(\overline{S_{j-1}}) \leq |\overline{S_{j-1}}|\, \textrm{ and } \, \FScomp{\widehat{k}}(S_j) \leq \frac{|S_j|}{|r_k|}. \label{PD: Thm: fs not PD: EQ3}
\end{align}

By Lemma \ref{FS: Lemma: LEQ Lemma}, whenever $\FScomp{\widehat{k}}(\overline{S_{j-1}})$ and  $\FScomp{\widehat{k}}(S_{j})$ are large (i.e. for long enough prefixes of S) we have that there exists a $k'$ such that \begin{align}
    \FScomp{k'}(\overline{S_j}) &\leq 2\FScomp{\widehat{k}}(\overline{S_{j-1}}) + \FScomp{\widehat{k}}(S_j) + 2 \tag{by Lemma \ref{FS: Lemma: LEQ Lemma}} \\
    & \leq 2|\overline{S_{j-1}}| + \frac{|S_j|}{|r_k|} + 2 \tag{by \eqref{PD: Thm: fs not PD: EQ3}} \\
    & = 2(\log(|S_j|) + 1) + \frac{|S_j|}{|r_k|}. \label{FS: Thm: io not ae deep: EQ4} 
\end{align}

Let $k$ be large. Using that $\lim\limits_{k \to \infty}(4k+1) / |r_k| = 0$, for sufficiently long prefixes of the form $\overline{S_j}$ where $k$ is devoted to $j$ it holds that
\begin{align}
    \FScomp{k}(\overline{S_j}) - \FScomp{k'}(\overline{S_j}) & \geq \frac{|S_j|}{|r_k|}(|r_k| - 4k -1) - 2(\log(|S_j|) + 1) \tag{by \eqref{PD: Thm: fs not PD: EQ1} and \eqref{FS: Thm: io not ae deep: EQ4}}\\
    & \geq |S_j|(1 - \frac{\alpha}{2}) \tag{for $j$ and $k$ sufficiently large} \\
    & = (|\overline{S_j}| -\log|S_j|)(1 - \frac{\alpha}{2}) \notag \\ 
    & \geq |\overline{S_j}|(1 - \alpha) \label{FS: Thm: io not ae deep: EQ5}
\end{align}
for $j$ sufficiently large. 

Similarly as the above only holds for large enough $k$, for all $i \leq k$ for $k$ large, when $k'$ is chosen as above we have that $\FScomp{i}(\overline{S_j}) - \FScomp{k'}(\overline{S_j}) \geq |\overline{S_j}|(1 - \alpha)$ when $j$ is devoted to $k$. Hence $\mathrm{FS}$-depth$(S) > (1-\alpha)$.

Next we show that PD-depth$(S) < \alpha$. Throughout the remainder of the proof we assume that $j$ is odd.

Let $C$ be any ILPDC. Consider the tuple $$(\overline{S_{j-1}},q_s,q_e,\widehat{C},\bar{\nu_C}(S_j))$$ where $q_s$ is the state $C$ when begins reading $S_j$, $q_e$ is the state $C$ ends up in after reading $S_j$, $\widehat{C}$ is an encoding of $C$ in some representation of PDCs and $\bar{\nu_C}(S_j)$ is the suffix of $C(\overline{S_j})$ outputted when reading $S_j$. Then given this tuple, one can recover $S_j$ as $C$ is information lossless.

Using the fact that tuples of the form $(x_1,x_2,\ldots,x_n)$ can be encoded by the string \begin{equation}1^{\lceil \log n_1 \rceil}0n_1x_11^{\lceil \log n_2 \rceil}0n_2x_2 \ldots1^{\lceil \log n_{n-1} \rceil}0n_{n-1}x_{n-1}x_n,\label{Tuple Encoding}\end{equation} where $n_i = |x_i|$ in binary, we have that for \begin{align}
    |S_j| \leq K(S_j) & \leq |\bar{\nu_C}(S_j)| + 2\log(|\overline{S_{j-1}}|) + |\overline{S_{j-1}}| + O(|\widehat{C}|) + O(1). \tag{as $j$ is odd} \\
    & = |\bar{\nu_C}(S_j)| + 2\log(\log(|S_j|)) + \log(|S_j|) + O(|\widehat{C}|) + O(1). \label{PD: Thm: fs not pd: EQ1}
\end{align}
Hence for $j$ large we have that
\begin{align}
    |\bar{\nu_C}(S_j)| \geq |S_j| - 2\log(\log(|S_j|)) - \log(|S_j|) - O(|\widehat{C}|) - O(1) > |S_j|(1 - \frac{\alpha}{2}). \label{PD: Thm: fs not pd: EQ2}
\end{align} Therefore, for $j$ large we have that  \begin{align}
    |C(\overline{S_j})| &\geq |\bar{\nu_C}(S_j)| 
     > |S_j|(1 - \frac{\alpha}{2}) \tag{by \eqref{PD: Thm: fs not pd: EQ2}} \\
    & = (|\overline{S_j}| - \log(|\overline{S_j}|))(1 - \frac{\alpha}{2}) > |\overline{S_j}|(1 - \alpha). \label{PD: Thm: fs not pd: EQ3}
\end{align}

Hence, for infinitely many prefixes $\overline{S_j}$ of $S$ it holds that \begin{equation}
|I_{\textrm{PD}}(\overline{S_j})| - |C(\overline{S_j})| < |\overline{S_j}| - |\overline{S_j}|(1 - \alpha) = \alpha|\overline{S_j}| \label{PD: Thm: fs not pd: EQ4}
\end{equation}As $C$ was chosen arbitrarily, it holds that PD-depth(S) $< \alpha$.

\end{proof}

The next result demonstrates the existence of a sequence which achieves a PD-depth of roughly $1/2$ while at the same time while having a small finite-state depth level. This sequence is split into chunks of strings of the form $RFR^{-1}$ where $F$ is a flag and $R$ is a string not containing $F$ with large plain Kolmogorov complexity relative to its length. A large ILPDC $C$ is built to push $R$ onto its stack, and then when it sees the flag $F$, uses its stack to compress $R^{-1}$. These $R$ are such that an $\ILUPDC$ cannot use its stack to compress $R$, resulting in no compression. For the finite-state transducers, the sequence appears almost random and so little depth is achieved. The sequence from Theorem $3$ of \cite{DBLP:journals/mst/MayordomoMP11} satisfies the theorem as shown.

\begin{theorem}
For all $0< \beta < 1/2$, there exists a sequence $S$ such that PD-depth$(S) \geq 1/2 - \beta$, and FS-depth$(S) < \beta$.
\label{PD: Thm: pd not fs}
\end{theorem}

\begin{proof}
Let $0 < \beta < 1/2$, and let $k> 8$ be a positive integer such that $\beta \geq 8/k$. For each $n$, let $t_n = k^{\lceil\frac{\log n}{\log k}\rceil}$. Note that for all $n$, \begin{align}
    n \leq t_n \leq kn. \label{PD: Thm: pd not fs3: EQ1} 
\end{align}
Consider the set $T_j$ which contains all strings of length $j$ that do not contain $1^{k}$ as a substring. As $T_j$ contains strings of the form $x_10x_20x_30\cdots$ where each $x_t$ is a string of length $k-1$, we have that $|T_j| \geq 2^{j(1 - \frac{1}{k})}$.
For each $j$, let $R_j \in T_{kt_j}$ have maximal plain Kolmogorov complexity in the sense that 
\begin{equation}
K(R_j) \geq |R_j|(1 - \frac{1}{k}). \label{PD: Thm: pd not fs3: EQ2}
\end{equation} Such an $R_j$ exists as $|T_{|R_j|}| > 2^{|R_j|(1 - \frac{1}{k})}-1$. Note that $kj \leq |R_j| \leq k^2j$. We construct $S$ in stages $S = S_1S_2\ldots$ where for each $j$, $$S_j = R_j1^kR_j^{-1}.$$

First we examine how well any ILUPDC compresses occurrences of $R_j$ zones in $S$. Let $C \in \ILUPDC$. Consider the tuple  $$(\widehat{C},q_s,q_e,z,\nu_C(q_s,R_j,z))$$
where $\widehat{C}$ is an encoding of $C$, $q_s$ is the state that $C$ begins reading $R_j$ in, $q_e$ is the state $C$ ends up in after reading $R_j$, $z$ is the stack contents of $C$ as it begins reading $R_j$ in $q_s$ (i.e. $z = 0^pz_0$ for some $p$), and the output $\nu_C(q_s,R_j,z)$ of $C$ on $R_j$. By Remark \ref{PD: Rmk: height discussion}, $C$'s stack is only important if $|z| < (c+1)|R_j|$, as if $|z|$ is larger, $C$ will output the same irregardless of $|z|$'s true value.  Hence, set \begin{equation}
    z' = \begin{cases}
         |z| & \text{ if } |z| < (c+1)|R_j| \\
         (c+1)|R_j| & \text{ if } |z| \geq (c+1)|R_j|
    \end{cases}
\end{equation}

As $C$ is lossless, having knowledge of the tuple $(\widehat{C},q_s,q_e,z',\nu_C(q_s,R_j,z))$ means we can recover $R_j$. If we encode the tuple $(\widehat{C},q_s,q_e,z',\nu_C(q_s,R_j,z))$ the same way as in \eqref{Tuple Encoding}, and noting that $z'$ contributes roughly $O(\log |R_j|)$ bits to the encoding, we have we have by Equation \eqref{PD: Thm: pd not fs3: EQ2} that \begin{equation}
    |R_j|(1 - \frac{1}{k}) \leq K(R_j) \leq |\nu_C(q_s,R_j,z)| + O(\log|R_j|) + O(|\widehat{C}|) + O(1). \label{PD: Thm: pd not fs3: EQ3}
\end{equation}
Therefore, for $j$ large we have \begin{align}
    |\nu_C(q_s,R_j,z)| \geq |R_j|(1 - \frac{1}{k}) - O(\log|R_j|) > |R_j|(1 - \frac{2}{k}). \label{PD: Thm: pd not fs3: EQ4}
\end{align}
This is similarly true for $R_j^{-1}$ zones also as $K(R_j) \leq K(R_j^{-1}) + O(1)$. Hence for $j$ large we see that $C$ outputs at least
\begin{align}
    |C(\overline{S_j})| - |C(\overline{S_{j-1}})| &\geq 2|R_j|(1 - \frac{2}{k}) \notag\\
    & = (|S_j| - k)(1 - \frac{2}{k}) \notag \\
    & \geq |S_j|(1 - \frac{3}{k}) \label{PD: Thm: pd not fs3: EQ5}
\end{align}
bits when reading $S_j$.

Next we examine how well $C$ compresses $S$ on arbitrary prefixes. Consider the prefix $S \upharpoonright n$ and let $j$ be such that $\overline{S_j}$ is a prefix of $S \upharpoonright n$ but $\overline{S_{j+1}}$ is not. Thus $S\upharpoonright n = \overline{S_j}\cdot y$ for some $y \sqsubset S_{j+1}$. Suppose Equation \eqref{PD: Thm: pd not fs3: EQ5} holds for all $i\geq j'$. Hence we have that
\begin{align}
    |C(S \upharpoonright n)| & \geq |C(\overline{S_j})| \notag \geq |C(\overline{S_j})| - |C(\overline{S_{j'-1}})| \notag \\
    & \geq |S_{j'}\ldots S_j|(1 - \frac{3}{k}) - O(1) \tag{by \eqref{PD: Thm: pd not fs3: EQ5}} \\
    & = (n - |y| - |\overline{S_{j'-1}}|(1 - \frac{3}{k}) - O(1) \notag \\
    & \geq (n - |y|)(1 - \frac{4}{k}). \label{PD: Thm: pd not fs3: EQ6}
\end{align}
Noting that $n = \Omega(j^2)$ and $|y| = O(j)$, by Equation \eqref{PD: Thm: pd not fs3: EQ6} we have that \begin{equation}
    |C(S \upharpoonright n)| \geq n(1 - \frac{5}{k}).\label{PD: Thm: pd not fs3: EQ7}
\end{equation}
As $C$ was arbitrary, we therefore have that \begin{equation}
    \rho_{\textrm{UPD}}(S) > 1 - \frac{6}{k}. \label{PD: Thm: pd not fs3: EQ8}
\end{equation}

Next we build an ILPDC $C'$ that is able to compress prefixes of $S$. In \cite{DBLP:journals/mst/MayordomoMP11}, it was shown that $R_{\textrm{PD}}(S) \leq 1/2$. We provide the details of this proof here for completeness.

For the ILPDC $C'$, informally, $C'$ outputs its input for some prefix $S_1\ldots S_i$. Then, for all $j > i$, $C'$ compresses $S_j$ as follows: On $S_j$, $C'$ outputs its input on $R_j1^k$ while trying to identify the $1^k$ flag. Once the flag is found, $C'$ pops the flag from its stack and then begins to read an $R_j^{-1}$ zone. On $R_j^{-1}$, $C'$ counts modulo $v$ to output a zero every $v$ bits, and uses its stack to ensure that the input is indeed $R_{j}^{-1}$. If this fails, $C'$ outputs an error flag, enters an error state and from then on outputs its input. Furthermore, $v$ is cleverly chosen such that for all but finitely many $j$, $v$ divides evenly into $|R_j|$. Specifically we set $v = k^a$ for some $a \in \N$. A complete description of $C'$ is provided at the end of this proof for completeness.

Next we will compute the compression ratio of $C'$ on $S$. We let $p$ be such that for all $j \geq p$, $v$ divides evenly into $|R_j|$. $C'$ will output its input on $\overline{S_{p-1}}$ and begin compressing on the succeeding zones. Also, note that the compression ratio of $C'$ on $S$ is largest on prefixes ending with a flag $1^k$. Hence, consider some prefix $\overline{S_{j-1}}R_j1^k$ of $S$. We have that for $j$ sufficiently large \begin{align}
    \frac{|C(\overline{S_{j-1}}R_j1^k)|}{|\overline{S_{j-1}}R_j1^k|} & \leq \frac{|\overline{S_{p-1}}| + \sum_{i = p}^j (|R_i| + k + \frac{|R_i|}{v}) - \frac{|R_j|}{k}}{|\overline{S_{j-1}}R_j1^k|} \notag \\
    & \leq \frac{|\overline{S_{p-1}}|}{|\overline{S_{j-1}}|} + \frac{(1 + \frac{1}{v})\sum_{i = 1}^jkt_i + jk - \frac{kt_j}{v}}{|\overline{S_{j-1}}|} \notag \\
    & \leq \frac{1}{6v} + \frac{(1 + \frac{1}{v})\sum_{i = 1}^jkt_i + jk - \frac{kt_j}{v}}{ 2k \sum_{i = 1}^{j-1}t_i} \tag{for $j$ large} \\
    & \leq \frac{1}{6v} + \frac{(1 + \frac{1}{v})\sum_{i = 1}^jt_i + j - \frac{t_j}{v}}{ 2 \sum_{i = 1}^{j-1}t_i} \notag\\
    & \leq \frac{1}{6v} + \frac{(1+ \frac{1}{v})\sum_{i = 1}^{j-1}t_i}{2 \sum_{i = 1}^{j-1}t_i} + \frac{t_j}{2 \sum_{i = 1}^{j-1}t_i} + \frac{j}{2 \sum_{i = 1}^{j-1}t_i} \notag\\
    & \leq \frac{1}{6v} + \frac{1}{2} + \frac{1}{2v} + \frac{1}{2}\Big(\frac{jk}{(j-1)(j)/2}\Big) + \frac{1}{2}\Big(\frac{j}{(j-1)(j)/2}\Big) \notag \\
    & \leq \frac{1}{6v} + \frac{1}{2} + \frac{1}{2v} + \frac{1}{6v} + \frac{1}{6v} \tag{for $j$ large} \\
    & = \frac{1}{2} + \frac{1}{v}.
\end{align}  

As $v$ can be chosen to be arbitrarily large, we therefore have that \begin{align}
    R_{\textrm{PD}}(S) \leq \frac{1}{2}. \label{PD: Thm: pd not fs3: EQ9}
\end{align}

Hence, for almost every $n$, by Equations \eqref{PD: Thm: pd not fs3: EQ8} and \eqref{PD: Thm: pd not fs3: EQ9} it follows that for all $C \in \ILUPDC$ \begin{align}
    |C(S \upharpoonright n)| - |C'(S \upharpoonright n)| & \geq (1 - \frac{6}{k} - \frac{1}{k})n - (\frac{1}{2} + \frac{1}{k})n \\
    & = (\frac{1}{2} - \frac{8}{k}). \label{PD: Thm: pd not fs3: EQ10}
\end{align} Choosing $k$ large such that $\frac{8}{k} \leq \beta$ gives us our desired result of PD-depth$(S) \geq \frac{1}{2} - \beta.$

Next we examine the finite-state depth of $S$. From Equation \eqref{PD: Thm: pd not fs3: EQ8} and Definition \ref{FS: Def: Dimension}, it follows that $\dimFS(S) > 1 - \frac{6}{k}$. Hence, for all $l$ it follows that for all but finitely many $n$ that \begin{equation}
    \FScomp{l}(S \upharpoonright n) \geq (1 - \frac{7}{k})n. \label{PD: Thm: pd not fs3: EQ11}
\end{equation}
Therefore, for $l'$ such that $I_{\textrm{FS}} \in \FSTsize{l'}$ we have that for all $l$ and almost every $n$  \begin{equation}
    \FScomp{l'}(S \upharpoonright n) - \FScomp{l}(S \upharpoonright n) \leq n - (1 - \frac{7}{k})n < \frac{8}{k}\cdot n. \label{PD: Thm: pd not fs3: EQ12}
\end{equation}
That is, FS-depth$(S) < \beta$ as desired.

For completeness we now present a full description of the ILPDC $C'$: Let $Q$ be the following set of states: \label{ILPDC Construction}
\begin{enumerate}
    \item the start state $q_0^s$,
    \item the counting states $q^s_1,\ldots q^s_m$ and $q_0$ that count up to $m = |\overline{S_{p-1}}|$,
    \item the flag checking states $q_1^{f_1},\ldots,q_k^{f_1}$ and $q_1^{f_0},\ldots,q_k^{f_0}$,
    \item the pop flag states $q_0^F,\ldots, q_k^F$,
    \item the compress states $q_1^c,\ldots, q_{v+1}^c$,
    \item the error state $q_e$.
\end{enumerate}

\noindent We now describe the transition function of $C'$. At first, $C'$ counts om $q_0^s$ to $q^s_m$ to ensure that for later $R_j$ zones, $v$ divides evenly into $|R_j|$. That is, for $0 \leq i \leq m-1$, $$\delta(q_i^s,x,y) = (q_{i+1}^s,y)$$ and $$\delta(q_m^s,\lambda,y) = (q_0, y).$$

\noindent Once this counting has taken place, an $R_j$ zone begins. Here, the input is pushed onto the stack and $C'$ tries to identify the flag $1^k$ by examining group of $k$ symbols. We set \begin{align*}
    \delta(q_0,x,y) = 
    \begin{cases}
         (q_1^{f_1},xy) & \textrm{if $x = 1$}\\
         (q_1^{f_0},xy) & \textrm{if $x \neq 1$}
    \end{cases}
\end{align*}
and for $1 \leq i \leq k-1$, $$\delta(q_i^{f_0},x,y) = (q_{i+1}^{f_0},xy)$$ and \begin{align*}
    \delta(q_i^{f_1},x,y) = 
    \begin{cases}
         (q_{i+1}^{f_1},xy) & \textrm{if $x = 1$} \\
         (q_{i+1}^{f_0},xy) & \textrm{if $x \neq1$.}
    \end{cases}
\end{align*}
If the flag $1^k$ is not detected after $k$ symbols, the test begins again. That is $$\delta(q_k^{f_0},\lambda,y) = (q_0,y).$$
If the flag is detected, the pop flag state is entered. $\delta(q_k^{f_1},\lambda,y) = (q_0^F,y).$ The flag is then removed from the stack, that is, for $0 \leq i \leq k$ $$\delta(q_i^F,\lambda,y) = (q_{i+1}^F, \lambda)$$ and $$\delta(q_k^F,\lambda,y) = (q_1^c,y).$$ 

\noindent $C'$ then checks using the stack, that the next part of the input it reads is $R_j^{-1}$, counting modulo $v$. If the checking fails, the error state is entered. That is for $1 \leq i \leq v$, \begin{align*}
    \delta(q_i^c,x,y) = 
    \begin{cases}
         (q_{i+1}^c,\lambda) & \textrm{if $x  = y$}\\
         (q_e, y) & \textrm{if $x \neq y$ and $y \neq z_0$} \\
         (q_1^{f_1},xz_0) & \textrm{if $x=1$ and $y = z_0$}\\
         (q_1^{f_0},xz_0) & \textrm{if $x\neq1$ and $y = z_0$}.
    \end{cases}
\end{align*}
Once $v$ symbols are checked, the checking starts again. That is $$\delta(q_{v+1}^c,\lambda,y) = (q_1^c,y).$$
The error state is the loop $$\delta(q_e,x,y) = (q_e,y).$$

We now describe the output function of $C'$. Firstly, on the counting states, $C'$ outputs its input. That is, for $0 \leq i \leq m-1$ $$\nu(q_i^s,x,y) = x.$$
On the flag checking states $C'$ outputs its input. That is, for $1 \leq i \leq k-1$ $$\nu(q_i^{f_0},x,y) = \nu(q_i^{f_1},x,y) = x.$$
$C'$ outputs nothing while in the flag popping states $q_0^F,\ldots, q_k^F$ and on the compression states $q_1^c,\ldots,q_{v+1}^c$ except in the case when $v$ symbols have just been checked. That is, $$\nu(q_{v}^c,x,y) = 0 \textrm{ if $x = y$}.$$
When an error is seen, a flag is outputted. That is for $1 \leq i \leq v$ $$\nu(q_i^c,x,y) = 1^{3m + i}0x \textrm{ if $x \neq y$ and $y \neq z_0$}.$$
$C'$ outputs its input while in the error state. That is, $$\nu(q_e,x,y) = x.$$

Lastly we verify that $C'$ is in fact IL. If the final state is not an error state, then all $R_j$ zones and $1^k$ flags are output as in the input. If the final state is $q_i^c$ then the number $t$ of zeros after the last flag in the output along with $q_i^c$ determines that the last $R_j^{-1}$ zone read is $tv + i-1$ bits long.
If the final state is $q_e$, then the output is of the form $$aR_j1^k0^t1^{3m + i}0b$$ for $a,b \in \fbins.$ The input is uniquely determined to be the input corresponding to the output $aR_j1^k0^t$ with final state $q_1^c$
 followed by $$R_j^{-1}[tv..tv + (i-1)-1].$$ As $1^{3m}$ does occur anywhere as a substring of $S$ post the prefix $\overline{S_{p-1}}$, its first occurrence post $\overline{S_{p-1}}$ as part of an output must correspond to an error flag.
 
 \end{proof}

\section{Lempel-Ziv Depth}
This section develops a notion of Lempel-Ziv depth (LZ-depth) based on the difference in compression of information lossless finite-state transducers and the Lempel-Ziv '78 (LZ) algorithm. We propose that LZ is a good choice to compare against ILFSTs as it asymptotically reaches the lower bound of compression attained by any finite-state compressor \cite{SheinwaldLZ,DBLP:LZ78}. Intuitively, a sequence is LZ-deep if given any $\ILFST$, the compression difference between the ILFST and the LZ algorithm is bounded below by a constant times the length of the prefix examined.

\begin{definition}
A sequence $S$ is \emph{(almost everywhere) Lempel-Ziv deep ((a.e.) LZ-deep)} if $$(\exists \alpha > 0)(\forall C \in \ILFST)(\forall^{\infty} n \in \N), |C(S \upharpoonright n)| - |\LZ(S \upharpoonright n)| \geq \alpha n.$$
\label{LZ: Def}
\end{definition}

\noindent We say a sequence is \emph{infinitely often (i.o.) LZ-deep} if the $(\forall^{\infty} n \in \N)$ term in the above definition is replaced with $(\exists^{\infty} n \in \N)$. We usually use the i.o. notation to denote i.o. LZ-depth but do not use the a.e. notation to denote a.e. LZ-depth. Therefore, one can assume LZ-depth always refers to the a.e. notion.

Prior to comparing LZ-depth with other notions, we require the following definition which describes the \textit{depth-level} of a sequence. This is similar to Definition \ref{PD: Def: Depth Level}.

\begin{definition}
Let $S \in \infbins$. Let $\alpha > 0$. We say that $\LZ\textit{-depth}(S) \geq \alpha$ if $$(\forall C \in \ILFST)(\forall^{\infty} n \in \N) \, |C(S \upharpoonright n)| - |\LZ(S \upharpoonright n)| \geq \alpha n.$$ Otherwise we say that $\LZ-\textit{depth}(S) < \alpha.$
\label{LZ: Def: Depth level}
 \end{definition}

We now show that LZ-depth satisfies the property of depth which says that trivial sequences and random sequences are not deep. Here, trivial means sequences which are FST-trivial (i.e. have a finite-state strong dimension of $0$ (recall Definition \ref{FS: Def: Dimension})) and random means LZ-incompressible.

\begin{theorem}
Let $S \in \infbins$.
\begin{enumerate}
    \item If $\rho_{\mathrm{LZ}}(S) = 1$, then $S$ is not $\LZ$-deep.
    \item If $\DimFS(S) = 0$, then $S$ is not $\LZ$-deep.
\end{enumerate}
\label{LZ: Thm: Easy Hard}
\end{theorem}

\begin{proof}
The proof follows the same structure as Theorem \ref{PD: Thm: Easy Hard}.

Let $S \in \infbins$ be such that $\rho_{\mathrm{LZ}}(S) = 1$. Therefore for every $\alpha > 0$, for almost every $n$ \begin{equation}|\LZ(S \upharpoonright n)| > n(1 - \alpha). \label{LZ: Thm: Easy Hard: EQ1}\end{equation}

\noindent Then for almost every $n$
\begin{equation}|I_{\mathrm{FS}}(S \upharpoonright n)| - |\LZ(S \upharpoonright n)| < n - n(1-\alpha) = \alpha n.\label{LZ: Thm: Easy Hard: EQ2}\end{equation} As $\alpha$ is arbitrary and $I_{\mathrm{FS}} \in \ILFST$, $S$ is not LZ-deep.

Next suppose $S \in \infbins$ is such that $\DimFS(S) = 0$. Hence there exists some $C \in \ILFST$ such that for every $\beta > 0$ and almost every $n$, \begin{equation}|C(S \upharpoonright n)| < \beta n.\label{LZ: Thm: Easy Hard: EQ3}\end{equation}  

\noindent Therefore it holds that for almost every $n$ \begin{equation}
    |C(S \upharpoonright n)| - |\LZ(S \upharpoonright n)| \leq |C(S \upharpoonright n)| < \beta n.\label{LZ: Thm: Easy Hard: EQ4}\end{equation}
As $\beta$ is arbitrary, $S$ is not LZ-deep.

\end{proof}

\subsection{Separation from Finite-State Depth}

 In our first result we demonstrate the existence of a LZ-deep sequence which is not i.o FS-deep. It relies on a result by Lathrop and Strauss which demonstrates the existence of a normal sequence $S$ such that $R_{\LZ}(S) < 1$, i.e. a normal sequence Lempel-Ziv can somewhat compress \cite{maliciousLZ}. The proof relies on the famous result that for every sequence $S$ which is normal, it holds that $\rho_{\ILFST}(S) = 1$ \cite{DBLP:journals/tcs/BecherH13}.

\begin{theorem}
There exists a normal $\LZ$-deep sequence. \label{LZ: Thm: LZ normal}
\end{theorem}

\begin{proof}
Let $S$ be the normal sequence from Theorem 4.3 of \cite{maliciousLZ} such that $R_{\LZ}(S) = \varepsilon <  1.$ Let $\delta>0$ be such that $\varepsilon + \delta < 1.$ Therefore, for almost every $n$ it holds that \begin{equation}
    |\LZ(S \upharpoonright n)| \leq (\varepsilon + \frac{\delta}{2})n.
    \label{LZ: Thm: LZ normal: EQ1}
\end{equation}

As $S$ is normal, $\rho_{\ILFST}(S) = 1$. Hence for all $C \in \ILFST$ and almost every $n$ we have that \begin{equation}
    |C(S \upharpoonright n)| \geq (1 - \frac{\delta}{2})n. \label{LZ: Thm: LZ normal: EQ2}
\end{equation}

Therefore for almost every $n$ \begin{align}
    |C(S \upharpoonright n)| - |\LZ(S \upharpoonright n)| \geq (1 - \frac{\delta}{2})n-(\varepsilon + \frac{\delta}{2})n = (1 - \varepsilon - \delta)n. \label{LZ: Thm: LZ normal: EQ3}
\end{align}
Therefore, as $C$ was arbitrary, $S$ is LZ-deep.

\end{proof}

  \begin{remark}Moser and Doty show that normal sequences are not FS-deep in \cite{DotyM07}. Hence the sequence satisfying Theorem \ref{LZ: Thm: LZ normal} is an example of an LZ-deep but not FS-deep sequence.
  \end{remark}
  
  Next we demonstrate that the sequence which satisfies Theorem \ref{PD: Thm: fs not PD} is an FS-deep sequence that is not LZ-deep. The long sections of randoms strings in $S$ prevent LZ-depth as was the case with PD-depth.

\begin{theorem}
There exists a sequence $S$ which is FS-deep but not LZ-deep.
\label{LZ: Thm: io Fs not LZ}
\end{theorem}

\begin{proof}
Let $0 < \beta < 1$. Let $S$ be the sequence satisfying Theorem \ref{PD: Thm: fs not PD}. Then FS-depth$(S) \geq \beta$, i.e. $S$ is FS-deep.

Next we show $S$ is not LZ-deep. Recall that $S$ is broken into substrings $S= S_1S_2S_3\ldots$ where $|S_1| = 2$ and for all $j$, $|S_j| = 2^{|S_1 \ldots S_{j-1}|}$, i.e. for $j> 1$, $\log|S_j| = |S_1\ldots S_{j-1}|$. Recall also that for $j$ odd, $S_j$ is a string of maximal plain Kolmogorov complexity in the sense that $K(S_j) \geq |S_j|$. We assume $j$ is always odd in the rest of the theorem. We write $\overline{S_j}$ to denote the prefix $S_1\ldots S_j$.

For any prefix of the form $\overline{S_j}$, as LZ is lossless, $\overline{S_j}$ can be recovered from the string $$d(\overline{S_{j-1}})\cdot01\cdot\LZ(S_j|\overline{S_{j-1}}).$$ Therefore for $j$ odd, \begin{equation}
    |S_j| \leq K(S_j) \leq 2|\overline{S_{j-1}}| + 2 + |\LZ(S_j|\overline{S_{j-1}})| + O(1),\label{LZ: Thm: io Fs not LZ: EQ1}\end{equation}
and so for $j$ large \begin{align}
    |\LZ(S_j|\overline{S_{j-1}} )| &\geq |S_j| - 2|\overline{S_{j-1}}| - O(1) \notag \\   
    & = |S_j| - 2\log(|S_j|) - O(1) > |S_j|(1 - \frac{\beta}{2}). \label{LZ: Thm: io Fs not LZ: EQ2}
    \end{align}
    Therefore for infinitely many prefixes of the form $\overline{S_j}$, we have
    \begin{align}
        |\LZ(\overline{S_j})| &\geq |\LZ(S_j | \overline{S_{j-1}})|  > |S_j|(1 - \frac{\beta}{2}) \tag{by \eqref{LZ: Thm: io Fs not LZ: EQ2}} \\
        & = (|\overline{S_j}| - \log(|S_j|))(1 - \frac{\beta}{2}) \notag \\
        & > |\overline{S_j}|(1 - \frac{\beta}{2})(1 - \frac{\beta}{2}) > |\overline{S_j}|(1 - \beta). \label{LZ: Thm: io Fs not LZ: EQ3}
    \end{align}
    Hence for infinitely many prefixes of $S$ we have that \begin{align}
    |I_{\mathrm{FS}}(\overline{S_j})| - |\LZ(\overline{S_j})| < |\overline{S_j}| - |\overline{S_j}|(1 - \beta) = |\overline{S_j}|\beta. \label{LZ: Thm: io Fs not LZ: EQ4}
\end{align}
As $\beta$ was arbitrary, it follows that $S$ is not LZ-deep.

\end{proof}

The following result follows from the previous theorem and shows that it is possible to build i.o. LZ-deep sequences which are not a.e. LZ-deep. The sequence from Theorem \ref{PD: Thm: fs not PD}, while finite-state deep and not LZ-deep, it is in fact infinitely often LZ-deep. This is because the LZ algorithm is able to compress the sections of the sequence composed of repetitions of random strings.

\begin{lemma}
There exists a sequence $S$ which is i.o. LZ-deep and FS-deep but not a.e. LZ-deep.
\label{LZ: Lemma: io LZ not ae LZ}
\end{lemma}

\begin{proof}
Let $S$ be the sequence satisfying Theorems \ref{PD: Thm: fs not PD} and \ref{LZ: Thm: io Fs not LZ} which is i.o. FS-deep but not a.e. LZ-deep. All that remains to show is that $S$ is i.o. LZ-deep.

Recall to construct $S$, we split the non-negative integers into intervals $I_1,I_2,\ldots$ such that $|I_1| = 2$ and $|I_j| = 2^{|I_1| + \cdots |I_{j-1}|}$ for all $j > 1.$ Also recall that for all $k \geq 1$, $k$ was devoted to intervals $I_j$ where $j$ had the form $j = 2^k + t(2^{k+1})$, for $t \geq 0.$ Recall also that $S$ was built in stages $S = S_1S_2\ldots$ such that if $k$ was devoted to interval $I_j$ then $S_j = r_k^{|I_j|/|r_k|}$ where $r_k$ was a string of length $|I_{2^k}|$ that was $3k$-finite-state random in the sense that \begin{equation}
    \FScomp{3k}(r_k) \geq |r_k| - 4k. \label{LZ: Lemma: io LZ not ae LZ: EQ1}
\end{equation}

We first examine how well any ILFST compresses prefixes of $S$. Let $C \in \ILFST$ with states $Q_C = \{q_1,\ldots, q_p\}.$ We assume all states of $C$ are reachable from its start state. For all $1 \leq i \leq p$, we let $C_i$ denote the FST with the same states, transitions and outputs as $C$ but with start state $q_i$, i.e. for all $x \in \fbins$, $C_i(x) = \nu_C(q_i,x).$ As $C$ is an ILFST, so is $C_i$. Recall from our encodings of FSTs that therefore $C_i \in \FSTsize{3|C|}$, where $|C|$ is the length of the encoding for $C$.

Next let $d$ be such that $I_{\mathrm{FS}}\in \FSTsize{d}$. As $C_i$ is an ILFST, by Lemma \ref{FS: Lemma: ILFST Machine less/greater} there exists $d'$ such that for all $i$ and $x$ \begin{equation}
    \FScomp{d'}(C_i(x)) \leq \FScomp{d}(x). \label{LZ: Lemma: io LZ not ae LZ: EQ2}
\end{equation}

Let $0< \varepsilon < 1$. By Lemma \ref{FS: Lemma: ILFST Machine less/greater}, there exists an $l$ such that for all $i$ and almost every $x$, \begin{equation}
    \FScomp{l}(x) \leq (1 + \frac{\varepsilon}{3})\FScomp{d'}(C_i(x)) + O(1).
    \label{LZ: Lemma: io LZ not ae LZ: EQ3}
\end{equation}
Of our set of random strings $\{r_k\}_{k \geq 1}$, let $l' \geq l$ be such that $r_{l'}$ satisfies both Equation \eqref{LZ: Lemma: io LZ not ae LZ: EQ3} and $|r_{l'}| - 4l' \geq (1 - \varepsilon/3)|r_{l'}|$. Such an $l'$ must exist as $\{r_k\}_{k \in \N}$ is a set of strings of increasing length.

Hence we have that for all $i$ \begin{align}
    |r_{l'}|(1 - \frac{\varepsilon}{3}) & \leq |r_{l'}| - 4l' \leq \FScomp{3l'}(r_{l'}) \leq \FScomp{l}(r_{l'}) \tag{as $r_{l'}$ is $3l'$-FS random} \\
    & \leq (1 + \frac{\varepsilon}{3})\FScomp{d'}(C_i(r_{l'})) + O(1) \tag{by \eqref{LZ: Lemma: io LZ not ae LZ: EQ3}} \\
    & \leq \FScomp{d'}(C_i(r_{l'})) + \frac{\varepsilon}{3}\FScomp{d}(r_{l'}) + O(1) \tag{by \eqref{LZ: Lemma: io LZ not ae LZ: EQ2}} \\
    & \leq |C_i(r_{l'})| + \frac{\varepsilon}{3}|r_{l'}| + O(1). \label{LZ: Lemma: io LZ not ae LZ: EQ4}
\end{align}
Hence by Equation \eqref{LZ: Lemma: io LZ not ae LZ: EQ4} for all $i$, when $l'$ is chosen large \begin{align}
    |C_i(r_{l'})| \geq |r_{l'}|(1 - \frac{\varepsilon}{3}) - \frac{\varepsilon}{3}|r_{l'}| - O(1) > |r_{l'}|(1 - \varepsilon). \label{LZ: Lemma: io LZ not ae LZ: EQ5}
\end{align}
That is, for all $i$ \begin{align}|C(q_i,r_{l'})| > |r_{l'}|(1 - \varepsilon).\label{LZ: Lemma: io LZ not ae LZ: EQ6}\end{align}

We now calculate a lower bound of compression of $C$ for prefixes of $S$ of the form $\overline{S_j}$ where $j$ is devoted to $r_{l'}$. For almost every such $j$ we have that \begin{align}
    |C(\overline{S_j})| &\geq |C(\overline{S_j})| - |C(\overline{S_{j-1}})| \notag\\
    & > \frac{|S_j|}{|r_{l'}|}(|r_{l'}|(1 - \varepsilon)) \tag{by \eqref{LZ: Lemma: io LZ not ae LZ: EQ6}} \\
    & = |S_j|(1 - \varepsilon) = (|\overline{S_j}| - |\overline{S_{j-1}}|)(1 - \varepsilon) \\
    & = (|\overline{S_j}| - \log(|S_j|))(1 - \varepsilon) \notag \\
    & > |\overline{S_j}|(1 - \beta) \label{LZ: Lemma: io LZ not ae LZ: EQ7}
\end{align}
where $\varepsilon < \beta < 1.$

We next examine how well LZ compresses any prefix $\overline{S_j}$ of $S$ where $j$ is devoted to $l'$. Note after reading $\overline{S_{j-1}}$, LZ's dictionary will contain at most $|\overline{S_{j-1}}|$ entries, i.e. it has size bounded above by $\log(|S_j|).$ Setting $a_{l'} = |r_{l'}| + 1$, by Lemma \ref{LZ: Lemma: Repeat lemma} we have that \begin{align}
    |\LZ(S_j|\overline{S_{j-1}})| &\leq \sqrt{2a_{l'}|S_j|}\log(|\overline{S_{j-1}}| + \sqrt{2a_{l'}|S_j|} \notag\\
    & = \sqrt{2a_{l'}|S_j|}\log(\log|S_j| + \sqrt{2a_{l'}|S_j|}) \notag\\
    & = O(\sqrt{|S_j|}\log|S_j|).\label{LZ: Lemma: io LZ not ae LZ: EQ8}
\end{align}

Hence we have that for $j$ large devoted to $l'$ that \begin{align}
    |\LZ(\overline{S_j})| &= |\LZ(\overline{S_{j-1}})| +  |\LZ(S_j|\overline{S_{j-1}})| \notag \\
    & \leq |\overline{S_{j-1}}| + o(|\overline{S_{j-1}}|) +  O(\sqrt{|S_j|}\log|S_j|) \tag{by \eqref{LZ: Lemma: io LZ not ae LZ: EQ8}}\\
    & = \log(|S_j|) + o(\log(|S_j|)) + O(\sqrt{|S_j|}\log|S_j|) \notag \\
    & = O(\sqrt{|S_j|}\log|S_j|). \label{LZ: Lemma: io LZ not ae LZ: EQ9}
\end{align}

Hence, as infinitely many intervals are devoted to $l'$, for infinitely many prefixes we have
\begin{align}
    |C(\overline{S_j})| - |\LZ(\overline{S_j})| & \geq |\overline{S_j}|(1 - \beta) - O(\sqrt{|S_j|}\log|S_j|) \tag{by \eqref{LZ: Lemma: io LZ not ae LZ: EQ7} and \eqref{LZ: Lemma: io LZ not ae LZ: EQ9}} \\
    & > |\overline{S_j}|( 1 - \alpha) \label{LZ: Lemma: io LZ not ae LZ: EQ10}
\end{align}
where $\beta < \alpha < 1.$ Hence as $C$ was an arbitrary ILFST, it holds that $S$ is i.o. LZ-deep.

\end{proof}

\subsection{Separation from Pushdown Depth}

The following subsection demonstrates the difference between LZ-depth with pushdown depth. We first demonstrate the existence of a sequence that has high LZ-depth but is not pushdown deep. We also show that we can build sequences that have PD-depth level of roughly $1/2$ which have a small LZ-depth level.

We first show demonstrate the existence of a LZ-deep sequence which is not pushdown deep. It relies on the construction of a sequence $S$ from Theorem $1$ of \cite{DBLP:journals/mst/MayordomoMP11} such that $R_{\LZ}(S) = 0$ but $\rho_{\textrm{PD}}(S) = 1.$ $S$ is constructed such that in contains repeated non-consecutive random substrings which Lempel-Ziv can exploit to compress. However, the random strings are long enough that a pushdown compressor cannot compress the sequence.

\begin{theorem}
\label{LZ: Thm: LZ not PD}
\end{theorem}

\begin{proof}
Let $S$ be the sequence from Theorem 1 of \cite{DBLP:journals/mst/MayordomoMP11} which satisfies $R_{\LZ}(S) = 0$ but $\rho_{\textrm{PD}}(S) = 1.$ We omit a description of the construction of $S$ in this proof. As $\rho_{\textrm{PD}}(S) = 1$, $S$ is not PD-deep by Theorem \ref{PD: Thm: Easy Hard}.

Let $0 < \alpha < 1$. As $\rho_{\textrm{PD}}(S) = 1,$ it also holds that $\dimFS(S) = 1.$ Hence for all $\widehat{C}\in \ILFST$ it holds that 
\begin{equation}
    |\widehat{C}(S \upharpoonright n)| > (1 - \frac{\alpha}{2})n. \label{LZ: Thm: LZ not PD: EQ1}
\end{equation}
As $R_{\LZ}(S) = 0,$ it holds that for almost all $n$ that \begin{equation}
    |\LZ(S \upharpoonright n)| \leq \frac{\alpha}{2}n. \label{LZ: Thm: LZ not PD: EQ3}
\end{equation}
Hence for almost all $n$ we have that \begin{equation}
    |\widehat{C}(S \upharpoonright n)| - |\LZ(S \upharpoonright n)| \geq (1 - \frac{\alpha}{2})n - \frac{\alpha}{2}n = (1 - \alpha)n. \label{LZ: Thm: LZ not PD: EQ2}
\end{equation}
As $\widehat{C}$ was arbitrary, it holds that $S$ is LZ-deep.

\end{proof}

Next we demonstrate the existence of a sequence that roughly has a PD-depth level of $1/2$ while having a very small LZ-depth level. This sequence was first presented in Theorem $5$ of \cite{DBLP:journals/mst/MayordomoMP11} and was built by enumerating strings in such a way so that a pushdown compressor can use its stack to compress, but an ILUPDC cannot use their stacks due to their unary nature. LZ cannot compress the sequence either as it is similar to a listing of all strings by order of length (i.e. all strings of length $1$ followed by all strings of length $2$ and so on). LZ performs poorly on such sequences.

\begin{theorem}
For all $0 < \beta < 1/2$, there exists a sequence $S$ such that PD-depth$(S) \geq (1/2-\beta)$ but LZ-depth$(S) < \beta$.
\label{LZ: Thm: PD not LZ}
\end{theorem}

\begin{proof}
Let $0 < \beta < 1/2$.
We first give a brief description of the sequence from Theorem 5 of \cite{DBLP:journals/mst/MayordomoMP11} that satisfies the result.
Let $\varepsilon$ be such that $0<\varepsilon < \beta$, and let $k$ and $v$ be non-negative integers to be determined later.

For any $n \in \N$, let $T_n$ denote the set of strings of length $n$ that do not contain the substring $1^j$ in $x$ for all $j \geq k$. As $T_n$ contains the set of strings whose every $k^{\text{th}}$ bit is $0$, it follows that $|T_n| \geq 2^{(\frac{k-1}{k})n}$. Note that for every $x \in T_n$, there exists $y \in T_{n-1}$ and $b \in \bin$ such that $x = yb$. Hence 
\begin{equation}
|T_n| < 2|T_{n-1}|. \label{LZ: Thm: PD not LZ: EQ1}
\end{equation}

Let $A_n = \{a_{1_n},\ldots,a_{u_n}\}$ be the set of palindromes in $T_n$. As fixing the first $\lceil \frac{n}{2} \rceil$ bits determines a palindrome, $|A_n| \leq 2^{\lceil \frac{n}{2} \rceil}$. The remaining strings in $T_n - A_n$ are split into $v+1$ pairs of sets $X_{n,i} = \{x_{n,i,1},\ldots, x_{n,i,t_n^i}\}$ and $Y_{n,i} = \{y_{n,i,1},\ldots, y_{n,i,t_n^i}\}$ where $t_n^i = \lfloor\frac{|T_n - A_n|}{2v} \rfloor$ if $i \neq v+1$ and $$t_n^{v+1} = \frac{1}{2}(|T_n-A_n| - 2\sum_{i = 1}^v|X_{n,i}|), $$ $(x_{n,i,j})^{-1} = y_{n,i,j}$ for every $1 \leq j \leq t_n^i$ and $1 \leq i \leq v+1$ both $x_{n,i,1}$ and $y_{n,i,t_n}$ start with $0$ (excluding the case where both $X_{n,v+1}$ and $Y_{n,v+1}$ are the empty sets). Note that for convenience we write $X_i,Y_i$ for $X_{n,i},Y_{n,i}$ respectively.

$S$ is constructed in stages. Let $f(k) = 2k$ and $f(n+1) = f(n) + v + 2$. Note that $n < f(n) < n^2$ for large $n$. For $n \leq k-1$, $S_n$ is a concatenation of all strings of length $n$, i.e. $S_n = 0^n\cdot0^{n-1}1\cdots 1^{n-1}0 \cdot 1^n.$ For $n \geq k$,$$S_n = a_{1_n}\ldots a_{u_n}1^{f(n)}z_{n,1}z_{n,2}\ldots z_{n,v}z_{n,v+1}$$
    where $$z_{n,i} = x_{n,i,1}x_{n,i,2}\ldots x_{n,i,t^i_n-1}x_{n,i,t^i_n}1^{f(n) + i}y_{n,i,t^i_n}y_{n,i,t^i_n-1}\ldots y_{n,i,2}y_{n,i,1},$$ with the possibility that $z_{n,v+1} = 1^{f(n) + v + 1}$ only.
That is, $S_n$ is a concatenation of all strings in $A_n$ followed by a flag of $f(n)$ ones, followed by a concatenation of all strings in the $X_i$ zones and $Y_i$ zones separated by flags of increasing length such that each $Y_i$ zone is the $X_i$ zone written in reverse. Let $$S = S_1S_2\ldots S_{k-1}1^k1^{k+1}\ldots 1^{2k-1}S_kS_{k+1}\ldots$$ i.e. the concatenation of all $S_j$ zones with some extra flags between $S_{k-1}$ and $S_k$.

Then from \cite{DBLP:journals/mst/MayordomoMP11}, for $\varepsilon$ small, choosing $k$ and $v$ appropriately large we have that \begin{align}
    \rho_{LZ}(S) \geq 1 - \varepsilon, \textrm{ and } R_{\textrm{PD}}(S) \leq 1/2.
    \label{LZ: Thm: PD not LZ: EQ2}
    \end{align}

 Next we consider how any $C \in \ILUPDC$ performs on $S$. Let $\overline{S_j}$ denote the prefix $$S_1 \ldots S_{k-1}1^k\ldots 1^{2k -1}S_k\ldots S_j$$ of $S$ for all $j \geq k$. 

Suppose $C$ is reading zone $S_n$ and can perform at most $c$ $\lambda$-transitions in a row. We examine the proportion of strings in $T_n$ that give a large contribution to the output. For simplicity, we write $C(p,x,s) = (q,v)$ to represent that when $C$ is in state $p$ with stack contents $s$ ( i.e. $s = 0^az_0$ for some $a$), on input $x$, $C$ outputs $v$ and ends in state $q$, i.e. $C(p,x,s) =(\widehat{\delta_Q}(p,x,s),\widehat{\nu}(p,x,s))$. 

 For each $x \in T_n$, let
$$h_x = \min \{ |v| : \exists p,q \in Q, \exists s \in \{0^az_0 : a \geq 0\}, C(p,x,s) = (q,v)\}$$ be the minimum possible addition of the output that could result from reading $x$. We restrict ourselves to reachable combinations of pairs of states and choices for $s$. Let $$B_n = \{ x \in T_n : h_x \geq \frac{(k-2)n}{k} \} $$ be the incompressible strings that give a large contribution to the output.

Next consider $x' \in T_n - B_n$. There is a computation of $x'$ that results in $C$ outputting at most $\frac{(k-2)n}{k}$ bits. As $C$ is lossless, $x'$ can be associated uniquely to a start state $p_{x'}$, stack contents $s_{x'}$, end state $q_{x'}$ and output $v_{x'}$ where $|v_{x'}| < \frac{(k-2)n}{k}$ such that $C(p_{x'},x',s_{x'}) = (q_{x'},v_{x'}).$ Recall from our previous discussion that on reading $x$ of length $n$, if $C$ has a stack of height bigger than $(c+1)n$, the stack will have no impact on the compression of $x$. That is, if $k\neq k'$ but $k,k' \geq (c+1)n$, then $C(p_{x'},x',0^kz_0) = C(p_{x'},x',0^{k'}z_0).$

Hence we can build a map $g$ such that $g(x') = (p_{x'},s'_{x'},v_{x'},q_{x'})$ where $0 \leq s'_{x'} < (c+1)|x|$. As this map $g$ is injective, we can bound $|T_n - B_n|$ as follows.
\begin{align}
   | T_n - B_n| &\leq |Q|^2 \cdot (c+1)n \cdot 2^{< \frac{(k-2)n}{k}} \notag\\
    & < |Q|^2 \cdot (c+1)n \cdot 2^{\frac{(k-2)n}{k}}. \label{LZ: Thm: PD not LZ: EQ4}
\end{align}    

For $0 < \delta < 1/6$ whose value is determined later, as $|T_n| \geq 2^{\frac{(k-1)n}{k}}$, we have that for $n$ large
\begin{align}
    |B_n| &= |T_n| - |T_n - B_n| \notag\\
    & > |T_n| - |Q|^2 \cdot (c+1)n \cdot 2^{\frac{(k-2)n}{k}} \tag{by \eqref{LZ: Thm: PD not LZ: EQ4}} \\
    & > |T_n|(1 - \delta). \label{LZ: Thm: PD not LZ: EQ5} 
\end{align}

Similarly, as the flags are only comprised of $O(n^2)$ bits in each $S_n$ zone, we have for $n$ large that  
\begin{align}
|T_n|n > |S_n|(1 - \delta). \label{LZ: Thm: PD not LZ: EQ6}
\end{align}
Then for $n$ large (say for all $n\geq i$ such that \eqref{LZ: Thm: PD not LZ: EQ5} and \eqref{LZ: Thm: PD not LZ: EQ6} hold),

\begin{align}
    |C(\overline{S_n})|
    & > \frac{k-2}{k}\sum_{j=i}^n \sum_{x \in B_j}j = \frac{k-2}{k}\sum_{j=i}^m j |B_j| \notag\\
    & > \frac{k-2}{k}(1 - \delta)\sum_{j=i}^n j|T_j|  \tag{by \eqref{LZ: Thm: PD not LZ: EQ5}}\\
    & > \frac{k-2}{k}(1 - 2\delta)\sum_{j=i}^n |S_j| \tag{by \eqref{LZ: Thm: PD not LZ: EQ6}}\\
    & = \frac{k-2}{k}(1 - 2\delta)(|\overline{S_n}| - |\overline{S_{i-1}}|) \notag \\
    & > \frac{k-2}{k}(1 - 3\delta)|\overline{S_n}|. \label{LZ: Thm: PD not LZ: EQ7}
\end{align}

The compression ratio of $S$ on $C \in \ILUPDC$ is least on prefixes of the form $\overline{S_n}\cdot x_{n+1}$, where potentially $x_{n+1}$ is a concatenation of all the strings in $T_{n+1}-B_{n+1}$, i.e. the compressible strings of $T_{n+1}$. Let $x_{n+1}$ be a such a potential prefix of $S_{n+1}$. Then if $F_{n+1} = \sum_{i=0}^v (f(n+1) + i)$, the length of the flags in $S_{n+1}$, we can bound the length of $|x_{n+1}|$ as follows. For $n$ large we have

\begin{align}
    |x_{n+1}| &< |T_{n+1}-B_{n+1}|(n+1) + F_{n+1} \notag\\
    & < (|T_{n+1}| - |B_{n+1}|)(n+1) + (v+2)n^2 \notag \\
    &< \delta|T_{n+1}|(n+1) + \delta|T_{n}|(n+1) \tag{by \eqref{LZ: Thm: PD not LZ: EQ5}} \\
    & < 2\delta|T_n|(n+1) + \delta|T_{n}|(n+1)  \tag{by \eqref{LZ: Thm: PD not LZ: EQ1}}\\
    & = 3\delta|T_n|(n+1) \notag \\
    & <3\delta|S_1\ldots S_n|, \label{eq:8}
\end{align}
Hence for $n$ large
\begin{align}
    |C(\overline{S_n}x_{n+1})| & \geq (\frac{k-2}{k})(1 - 3\delta)(|\overline{S_n}x_{n+1}| - |x_{n+1}|) \notag\\
    & > (\frac{k-2}{k})(1 - 3\delta)(|\overline{S_n}x_{n+1}| - 3\delta |\overline{S_n}|) \tag{by \eqref{LZ: Thm: PD not LZ: EQ7}} \\
    & > (\frac{k-2}{k})(1 - 6\delta)|\overline{S_n}x_{n+1}|\notag  \\
    & > \frac{k-3}{k}|\overline{S_n}x_{n+1}| \label{LZ: Thm: PD not LZ: EQ8}
\end{align}
when $\delta$ is chosen sufficiently small, i.e. $\delta < \frac{1}{6(k-2)}$. Hence \begin{equation}
    \rho_{\textrm{UPD}}(S) \geq \frac{k-3}{k}. \label{LZ: Thm: PD not LZ: EQ9} \end{equation} Thus for all  $0< \varepsilon' < \frac{k-3}{k}$, for almost every $n$, $$|C(S \upharpoonright n)| \geq (\frac{k-3}{k} - \varepsilon')n.$$ Next let $\hat{C} \in \ILPDC$ be such that $\hat{C}$ achieves $R_{\textrm{PD}}(S) \leq 1/2.$ Then for all $\varepsilon' > 0$ and almost every $n$ it holds that $$|C(S \upharpoonright n)| \leq (\frac{1}{2} + \varepsilon').$$ Hence, for almost every $n$ and every $C \in \ILUPDC$ we have
\begin{align*}
    |C(S \upharpoonright n)| - |\hat{C}(S \upharpoonright n)| & \geq (\frac{k-3}{k} - \varepsilon')n - (\frac{1}{2} + \varepsilon')n \\
    & = (\frac{1}{2}- \frac{3}{k})n. \label{LZ: Thm: PD not LZ: EQ9}
\end{align*}
That is, PD-depth$(S) > \frac{1}{2}-\frac{3}{k}$. Hence, choosing $k$ large appropriately at the start such that $\frac{3}{k} < \beta$ we have that PD-depth$(S) > \frac{1}{2} - \beta.$

Next we examine LZ-depth. Recall $\rho_{LZ}(S) \geq 1 - \varepsilon$. Thus for $c$ such that $\varepsilon + c < \beta$ (recall $\varepsilon < \beta)$, for almost every $n$ it holds that \begin{equation}|LZ(S \upharpoonright n)| > (1 - \varepsilon-c)n. \label{LZ: Thm: PD not LZ: EQ10}\end{equation} Hence as $I_{\mathrm{FS}} \in \ILFST$, we have that for almost every $n$ \begin{equation}|I_{\mathrm{FS}}(S \upharpoonright n)| - |LZ(S \upharpoonright n)| < n - (1 - \varepsilon-c)n = (\varepsilon + c)n < \beta n. \label{LZ: Thm: PD not LZ: EQ11}\end{equation} Hence we have that LZ-depth$(S) < \beta$.

In conclusion, for all $0 < \beta < \frac{1}{2}$, choosing $\varepsilon$ such that $\varepsilon<\beta$ and $k$ such that $\frac{3}{k} < \beta$, a sequence $S$ can be built which satisfies the requirements of the theorem.

\end{proof}

  \bibliographystyle{plainurl}
\bibliography{biblio.bib}

\begin{thebibliography}{10}

\bibitem{DBLP:conf/birthday/AlbertMM17}
Pilar Albert, Elvira Mayordomo, and Philippe Moser.
\newblock Bounded pushdown dimension vs lempel ziv information density.
\newblock In {\em Computability and Complexity - Essays Dedicated to Rodney G.
  Downey on the Occasion of His 60th Birthday}, volume 10010 of {\em Lecture
  Notes in Computer Science}, pages 95--114. Springer, 2017.
\newblock \href {https://doi.org/10.1007/978-3-319-50062-1\_7}
  {\path{doi:10.1007/978-3-319-50062-1\_7}}.

\bibitem{b.antunes.depth.journal}
Luis Antunes, Lance Fortnow, Dieter van Melkebeek, and N.~V. Vinodchandran.
\newblock Computational depth: Concept and applications.
\newblock {\em Theor. Comput. Sci.}, 354(3):391--404, 2006.
\newblock \href {https://doi.org/10.1016/j.tcs.2005.11.033}
  {\path{doi:10.1016/j.tcs.2005.11.033}}.

\bibitem{DBLP:journals/siamcomp/AthreyaHLM07}
Krishna~B. Athreya, John~M. Hitchcock, Jack~H. Lutz, and Elvira Mayordomo.
\newblock Effective strong dimension in algorithmic information and
  computational complexity.
\newblock {\em {SIAM} J. Comput.}, 37(3):671--705, 2007.
\newblock \href {https://doi.org/10.1137/S0097539703446912}
  {\path{doi:10.1137/S0097539703446912}}.

\bibitem{DBLP:journals/jcss/BecherCH15}
Ver{\'{o}}nica Becher, Olivier Carton, and Pablo~Ariel Heiber.
\newblock Normality and automata.
\newblock {\em J. Comput. Syst. Sci.}, 81(8):1592--1613, 2015.
\newblock \href {https://doi.org/10.1016/j.jcss.2015.04.007}
  {\path{doi:10.1016/j.jcss.2015.04.007}}.

\bibitem{DBLP:journals/tcs/BecherH13}
Ver{\'{o}}nica Becher and Pablo~Ariel Heiber.
\newblock Normal numbers and finite automata.
\newblock {\em Theor. Comput. Sci.}, 477:109--116, 2013.
\newblock \href {https://doi.org/10.1016/j.tcs.2013.01.019}
  {\path{doi:10.1016/j.tcs.2013.01.019}}.

\bibitem{b:bennett88}
C.~H. Bennett.
\newblock Logical depth and physical complexity.
\newblock {\em The Universal Turing Machine, A Half-Century Survey}, pages
  227--257, 1988.

\bibitem{calude2011finite}
Cristian~S. Calude, Kai Salomaa, and Tania Roblot.
\newblock Finite state complexity.
\newblock {\em Theor. Comput. Sci.}, 412(41):5668--5677, 2011.
\newblock \href {https://doi.org/10.1016/j.tcs.2011.06.021}
  {\path{doi:10.1016/j.tcs.2011.06.021}}.

\bibitem{calude2016finite}
Cristian~S. Calude, Ludwig Staiger, and Frank Stephan.
\newblock Finite state incompressible infinite sequences.
\newblock {\em Inf. Comput.}, 247:23--36, 2016.
\newblock \href {https://doi.org/10.1016/j.ic.2015.11.003}
  {\path{doi:10.1016/j.ic.2015.11.003}}.

\bibitem{DBLP:journals/tcs/DaiLLM04}
Jack~Jie Dai, James~I. Lathrop, Jack~H. Lutz, and Elvira Mayordomo.
\newblock Finite-state dimension.
\newblock {\em Theor. Comput. Sci.}, 310(1-3):1--33, 2004.
\newblock \href {https://doi.org/10.1016/S0304-3975(03)00244-5}
  {\path{doi:10.1016/S0304-3975(03)00244-5}}.

\bibitem{DotyMoserLossy}
David Doty and Philippe Moser.
\newblock Finite-state dimension and lossy decompressors.
\newblock {\em CoRR}, 2006.
\newblock \href {http://arxiv.org/abs/cs/0609096} {\path{arXiv:cs/0609096}}.

\bibitem{DotyM07}
David Doty and Philippe Moser.
\newblock Feasible depth.
\newblock In {\em Computation and Logic in the Real World, Third Conference on
  Computability in Europe, CiE 2007, Siena, Italy, June 18-23, 2007,
  Proceedings}, volume 4497 of {\em Lecture Notes in Computer Science}, pages
  228--237. Springer, 2007.
\newblock \href {https://doi.org/10.1007/978-3-540-73001-9\_24}
  {\path{doi:10.1007/978-3-540-73001-9\_24}}.

\bibitem{DBLP:journals/tcs/DotyN07}
David Doty and Jared Nichols.
\newblock Pushdown dimension.
\newblock {\em Theor. Comput. Sci.}, 381(1-3):105--123, 2007.
\newblock \href {https://doi.org/10.1016/j.tcs.2007.04.005}
  {\path{doi:10.1016/j.tcs.2007.04.005}}.

\bibitem{DBLP:journals/tcs/DowneyMN17}
Rod Downey, Michael McInerney, and Keng~Meng Ng.
\newblock Lowness and logical depth.
\newblock {\em Theor. Comput. Sci.}, 702:23--33, 2017.
\newblock \href {https://doi.org/10.1016/j.tcs.2017.08.010}
  {\path{doi:10.1016/j.tcs.2017.08.010}}.

\bibitem{downey:book}
Rodney~G. Downey and Denis~R. Hirschfeldt.
\newblock {\em Algorithmic Randomness and Complexity}.
\newblock Springer, 2010.

\bibitem{borelNormal}
M.~\'{E}mile Borel.
\newblock Les probabilités dénombrables et leurs applications arithmétiques.
\newblock {\em Rendiconti del Circolo Matematico di Palermo}, 27(1):247--271,
  1909.
\newblock \href {https://doi.org/10.1007/BF03019651}
  {\path{doi:10.1007/BF03019651}}.

\bibitem{Huff59a}
D.~{Huffman}.
\newblock Canonical forms for information-lossless finite-state logical
  machines.
\newblock {\em IRE Transactions on Information Theory}, 5(5):41--59, 1959.
\newblock \href {https://doi.org/10.1109/TIT.1959.1057537}
  {\path{doi:10.1109/TIT.1959.1057537}}.

\bibitem{DBLP:conf/sofsem/JordonM20}
Liam Jordon and Philippe Moser.
\newblock On the difference between finite-state and pushdown depth.
\newblock In {\em 46th International Conference on Current Trends in Theory and
  Practice of Informatics, {SOFSEM} 2020, Limassol, Cyprus, January 20-24,
  2020, Proceedings}, volume 12011 of {\em Lecture Notes in Computer Science},
  pages 187--198. Springer, 2020.
\newblock \href {https://doi.org/10.1007/978-3-030-38919-2\_16}
  {\path{doi:10.1007/978-3-030-38919-2\_16}}.

\bibitem{DBLP:journals/tcs/JuedesLL94}
David~W. Juedes, James~I. Lathrop, and Jack~H. Lutz.
\newblock Computational depth and reducibility.
\newblock {\em Theor. Comput. Sci.}, 132(2):37--70, 1994.
\newblock \href {https://doi.org/10.1016/0304-3975(94)00014-X}
  {\path{doi:10.1016/0304-3975(94)00014-X}}.

\bibitem{Koha78}
Z.~Kohavi.
\newblock Switching and finite automata theory (second edition).
\newblock {\em McGraw-Hill}, 1978.

\bibitem{DBLP:journals/iandc/LathropL99}
James~I. Lathrop and Jack~H. Lutz.
\newblock Recursive computational depth.
\newblock {\em Inf. Comput.}, 153(1):139--172, 1999.
\newblock \href {https://doi.org/10.1006/inco.1999.2794}
  {\path{doi:10.1006/inco.1999.2794}}.

\bibitem{maliciousLZ}
James~I. Lathrop and Martin Strauss.
\newblock A universal upper bound on the performance of the lempel-ziv
  algorithm on maliciously-constructed data.
\newblock In {\em Compression and Complexity of {SEQUENCES} 1997, Positano,
  Amalfitan Coast, Salerno, Italy, June 11-13, 1997, Proceedings}, pages
  123--135. {IEEE}, 1997.
\newblock \href {https://doi.org/10.1109/SEQUEN.1997.666909}
  {\path{doi:10.1109/SEQUEN.1997.666909}}.

\bibitem{DBLP:journals/mst/MayordomoMP11}
Elvira Mayordomo, Philippe Moser, and Sylvain Perifel.
\newblock Polylog space compression, pushdown compression, and lempel-ziv are
  incomparable.
\newblock {\em Theory Comput. Syst.}, 48(4):731--766, 2011.
\newblock \href {https://doi.org/10.1007/s00224-010-9267-6}
  {\path{doi:10.1007/s00224-010-9267-6}}.

\bibitem{DBLP:journals/tcs/Moser13}
Philippe Moser.
\newblock On the polynomial depth of various sets of random strings.
\newblock {\em Theor. Comput. Sci.}, 477:96--108, 2013.
\newblock \href {https://doi.org/10.1016/j.tcs.2012.10.045}
  {\path{doi:10.1016/j.tcs.2012.10.045}}.

\bibitem{DBLP:journals/iandc/Moser20}
Philippe Moser.
\newblock Polylog depth, highness and lowness for {E}.
\newblock {\em Inf. Comput.}, 271:104483, 2020.
\newblock \href {https://doi.org/10.1016/j.ic.2019.104483}
  {\path{doi:10.1016/j.ic.2019.104483}}.

\bibitem{DBLP:journals/dmtcs/MoserS17}
Philippe Moser and Frank Stephan.
\newblock Depth, highness and {DNR} degrees.
\newblock {\em Discret. Math. Theor. Comput. Sci.}, 19(4), 2017.
\newblock \href {https://doi.org/10.23638/DMTCS-19-4-2}
  {\path{doi:10.23638/DMTCS-19-4-2}}.

\bibitem{nies:book}
Andr{\'e} Nies.
\newblock {\em Computability and Randomness}.
\newblock Oxford University Press, 2009.

\bibitem{SheinwaldLZ}
Dafna {Sheinwald}.
\newblock On the {Z}iv-{L}empel proof and related topics.
\newblock {\em Proceedings of the IEEE}, 82(6):866--871, 1994.
\newblock \href {https://doi.org/10.1109/5.286190}
  {\path{doi:10.1109/5.286190}}.

\bibitem{DBLP:journals/iandc/SheinwaldLZ95}
Dafna Sheinwald, Abraham Lempel, and Jacob Ziv.
\newblock On encoding and decoding with two-way head machines.
\newblock {\em Inf. Comput.}, 116(1):128--133, 1995.
\newblock \href {https://doi.org/10.1006/inco.1995.1009}
  {\path{doi:10.1006/inco.1995.1009}}.

\bibitem{DBLP:LZ78}
Jacob Ziv and Abraham Lempel.
\newblock Compression of individual sequences via variable-rate coding.
\newblock {\em {IEEE} Trans. Inf. Theory}, 24(5):530--536, 1978.
\newblock \href {https://doi.org/10.1109/TIT.1978.1055934}
  {\path{doi:10.1109/TIT.1978.1055934}}.

\end{thebibliography}
\end{document}